\documentclass[lettersize,journal]{IEEEtran}

\usepackage{latexsym,bm}
\usepackage{amsmath,amssymb,amsfonts}
\usepackage{amsthm}
\usepackage{array}

\usepackage{graphicx}
\usepackage{subfigure} % currently used in sections/5results.tex
\usepackage{booktabs}
\usepackage{multirow}
\usepackage{tabularx}
\usepackage{colortbl}
\usepackage[table,dvipsnames]{xcolor}
\usepackage{arydshln}
\usepackage{lscape}

\usepackage[ruled,linesnumbered,lined,noend]{algorithm2e}

\usepackage[switch]{lineno}
\usepackage{cite}
\usepackage{url}
\usepackage{comment}
\usepackage{balance}
\usepackage{ragged2e}
\usepackage{xspace}
\usepackage{ulem}
\usepackage{enumitem}
\usepackage{textcomp}

\usepackage{listings}
\usepackage{fvextra}
\usepackage{tcolorbox}
\usepackage{pifont}
\usepackage{float}

\usepackage[
  colorlinks,
  linkcolor=blue,
  citecolor=blue,
  filecolor=blue,
  urlcolor=blue
]{hyperref}

\newtheorem{theorem}{Theorem}
\newtheorem{proposition}{Proposition}

\definecolor{light-gray}{gray}{0.95}
\definecolor{dkgreen}{rgb}{0,0.6,0}
\definecolor{gray}{rgb}{0.5,0.5,0.5}
\definecolor{mauve}{rgb}{0.58,0,0.82}

\newcommand{\tool}{\textsf{EAHC}\xspace}
\newcommand{\code}[1]{\colorbox{light-gray}{\texttt{#1}}}

\newcommand{\yg}[1]{\textcolor{black}{#1}}

\newcommand{\reasonhalluc}{reasoning hallucination\xspace}
\newcommand{\codererank}{code reranking\xspace}
\newcommand{\funcequivcluster}{\yg{execution-equivalence cluster}\xspace}

\newcommand{\dualchannel}{dual-channel\xspace}
\newcommand{\execchannel}{execution channel\xspace}
\newcommand{\reasonchannel}{reasoning channel\xspace}

\newcommand{\hierselect}{hierarchical selection\xspace}

\newcommand{\compilerinloop}{compiler-in-the-loop\xspace}

\newcommand{\execsignal}{\textit{execution signal}\xspace}
\newcommand{\reasonsignal}{\textit{reasoning signal}\xspace}
\newcommand{\execvec}{\mathbf{e}}
\newcommand{\execscore}{S_{\text{exec}}}
\newcommand{\reasonscore}{S_{\text{reason}}}
\newcommand{\hybridscore}{S_{\text{hybrid}}}

\newcommand{\testbench}{testbench\xspace}
\newcommand{\temporallogic}{temporal logic\xspace}
\newcommand{\parallelexec}{parallel execution\xspace}
\newcommand{\signalprop}{signal propagation\xspace}

\begin{document}

\title{Execution-Anchored Hallucination Calibration Reranking for Verilog Code Generation}

\author{Guang Yang, Xing Hu\IEEEauthorrefmark{1}\thanks{Corresponding author: Xing Hu.}, Xiang Chen, Terry Yue Zhuo, and Xin Xia

\IEEEcompsocitemizethanks{
\IEEEcompsocthanksitem Guang Yang is with the State Key Laboratory of Blockchain and Data Security, Zhejiang University, Hangzhou, China, and also with the Hangzhou High-Tech Zone (Binjiang) Institute of Blockchain and Data Security, Hangzhou, China. 
Xing Hu and Xin Xia are with the State Key Laboratory of Blockchain and Data Security, Zhejiang University, Hangzhou, China. 
Xiang Chen is with the School of Artificial Intelligence and Computer Science, Nantong University, Nantong, China. 
Terry Yue Zhuo is with the Department of Computer Science, Monash University and CSIRO's Data61, Australia.
E-mail: novelyg@outlook.com, xinghu@zju.edu.cn, xchencs@ntu.edu.cn, terry.zhuo@monash.edu, xin.xia@acm.org.
}

\thanks{Manuscript received April 19, 2020; revised August xx, xxxx.}}

\markboth{IEEE TRANSACTIONS ON Software Engineering,~Vol.~XX, No.~XX, XX~2026}%
{Execution-Anchored Hallucination Calibration Reranking for Verilog Code Generation}

\IEEEtitleabstractindextext{
\begin{abstract}
\justifying
Large Language Models (LLMs) have demonstrated remarkable capabilities in code generation, yet their performance degrades significantly on low-resource Hardware Description Languages such as Verilog. 
While multi-candidate sampling improves the likelihood of generating correct solutions, automatically selecting the optimal candidate remains an open challenge.
Through a systematic empirical study across nine models and two benchmarks, we identify two critical limitations: (1) existing execution-based reranking methods, which rely on testbench pass/fail outcomes, exhibit poor domain transferability due to low-quality generated testbenches; and (2) LLM-as-a-Judge suffers from \textit{reasoning hallucination}, producing inconsistent judgments for \yg{execution-equivalent} code.
These findings reveal two signal types with orthogonal errors: execution signals (deterministic but testbench coverage limited) and reasoning signals (semantically rich but hallucination-prone).
\yg{Their orthogonality suggests combining the two signals, yet in our experiments letting the reasoner directly observe execution results merely anchors its judgments on test outcomes; we therefore acquire the two signals independently and fuse them only at the decision stage.}
Based on these insights, we propose {\tool}, an \textbf{E}xecution-\textbf{A}nchored \textbf{H}allucination \textbf{C}alibration reranking framework \yg{that anchors reasoning judgments to execution behavior so that execution-equivalent candidates receive consistent scores}, which implements a dual-channel architecture: {\tool}-R, a 4B reasoning discriminator fine-tuned on 47K compiler-verified judgment traces via multi-teacher distillation; and {\tool}-T, a testbench generator leveraging RAG over a 53K corpus for execution verification.
Experiments show that {\tool} \yg{ranks first in 15 of 18 configurations and attains the best average on both benchmarks}, elevating average Pass@1 from 53.99\% to \textbf{65.10\%} on VerilogEval-v2 and from 53.18\% to \textbf{68.25\%} on ResBench, recovering over 60\% of the gap to Pass@10 oracle.
\justifying
\end{abstract}

}

\maketitle

\IEEEdisplaynontitleabstractindextext
\IEEEpeerreviewmaketitle

% \linenumbers 
\section{Introduction}
\label{sec:intro}

% Para 1: 现实世界重要性 - HDL家族 → Verilog地位 → 代码生成重要性 → LLM迁移困境
Hardware Description Languages (HDLs), including Verilog, VHDL, and SystemVerilog, form the foundation of modern digital circuit design~\cite{flake2020verilog}.
Among these, Verilog stands out for its widespread industry adoption, serving as the standard in both ASIC and FPGA design workflows.
As integrated circuit complexity continues to grow, automated Verilog code generation has become increasingly critical to reduce development costs and alleviate engineering workload~\cite{yang2025large}.
With the rapid advancement of Large Language Models (LLMs), remarkable progress has been achieved in code generation for \textit{high-resource} programming languages, with state-of-the-art LLMs now surpassing 90\% Pass@1 on Python benchmarks like HumanEval~\cite{chen2021evaluatinglargelanguagemodels,du2024evaluating}.
Inspired by this success, researchers have attempted to transfer LLM capabilities to Verilog generation~\cite{joel2024survey}.
However, these efforts reveal a significant performance gap: even the most capable LLMs achieve only 30--50\% Pass@1 on Verilog benchmarks, far below their performance on general-purpose languages~\cite{liu2023verilogeval,thakur2024verigen}.

The performance gap stems from fundamental challenges: (1) \textit{data scarcity}, as available Verilog corpora are orders of magnitude smaller than those for general-purpose languages and (2) \textit{domain complexity}, since hardware design involves \temporallogic, \parallelexec, and \signalprop that diverge fundamentally from sequential software paradigms.
% While a natural solution is to build larger high-quality Verilog corpora, this approach is hard to scale: labeling costs grow rapidly with corpus size, and retraining large models for each domain remains costly.
A straightforward solution is to build large-scale, high-quality Verilog corpora.
However, it is limited by the escalating labeling cost and the high computational expense of retraining large models for this domain.

% Para 2: 量化Gap - 定义问题
An alternative strategy leverages the characteristics of LLM decoding: by sampling $k$ candidates at non-zero temperature, the probability that \textit{at least one} is correct (Pass@$k$) far exceeds that of a single greedy attempt (Pass@1).
Table~\ref{tab:empirical_results} reveals a striking disparity: while Pass@10 reaches 71.30\% on VerilogEval-v2~\cite{liu2023verilogeval} and 77.78\% on ResBench~\cite{guo2025resbench}, Pass@1 lags significantly at 53.99\% and 53.18\%, respectively. 
% This \textbf{17--25 percentage-point gap} represents considerable latent capability that existing approaches fail to exploit.
% However, this potential remains largely theoretical in practice.
% Engineering deployments demand a \textit{single deterministic solution}, not $k$ uncertain alternatives; manual selection is impractical as human effort scales linearly with $k$; and no effective automated method exists to identify the optimal candidate.
This \textbf{17--25 percentage-point gap} reveals considerable latent capability that existing approaches fail to fully leverage in practice. 
In real-world development scenarios, developers typically expect definitive code suggestions rather than $k$ uncertain alternatives, which require additional manual effort to evaluate and select.
This gap motivates the \textit{\codererank} problem: given $k$ candidates $\mathcal{Y}_k = \{\hat{y}_1, \ldots, \hat{y}_k\}$ for requirement $x$, design a scoring function to select the most likely correct implementation, effectively converting sampling diversity into deployment-ready accuracy.

% Para 3: Empirical Study - 两个关键发现
To better understand the limitations of existing approaches, we conduct a systematic empirical study (Section~\ref{sec:empirical}) across nine code generation models and two benchmarks, evaluating state-of-the-art reranking methods including generation probability~\cite{zhang2023av}, semantic matching~\cite{inala2022fault}, execution verification~\cite{chen2022codet}, and LLM-as-a-Judge~\cite{yang2025code}.
Our study reveals two critical findings:

\ding{73}\textbf{Finding 1: Poor Domain Transferability.}
Existing reranking methods for general-purpose languages generalize poorly to Verilog.
Probability-based approaches suffer from distributional shift between training corpora and HDL syntax, while semantic matching methods fail to capture hardware-specific constructs such as \code{always} blocks and non-blocking assignments.
Execution-based methods like CodeT~\cite{chen2022codet} harness \execsignal{}s for candidate selection, but their reliance on self-generated \testbench{}s limits effectiveness: the underlying code models lack hardware domain expertise to produce high-quality test cases.

\ding{73}\textbf{Finding 2: Reasoning Hallucination in LLM-as-a-Judge.}
In contrast to execution-based methods, LLM-as-a-Judge leverages \reasonsignal{}s and emerges as the most competitive baseline, yet it exhibits a critical flaw: \textit{\reasonhalluc}, where it may produce inconsistent correctness judgments for \yg{\textit{execution-equivalent} code, i.e., candidates with identical input--output behavior on a finite test suite (Section~\ref{sec:empirical})}.
Specifically, for candidates $\hat{y}_i$ and $\hat{y}_j$ with the same execution behavior, the LLMs may give different judgments.
This inconsistency stems from the fact that LLMs reason at the \textit{token} level without grounding in \textit{execution semantics}~\cite{wang2025open}.

% Para 4: Key Insight - 执行锚定
These findings reveal two signal types with distinct strengths and limitations.
Execution-based methods (Finding 1) leverage \textbf{\execsignal{}s}, i.e., deterministic pass/fail outcomes where \yg{execution-equivalent} code must produce identical results. 
Their strength lies in \textit{consistency}: identical behavior guarantees identical scores. 
However, they suffer from \textit{coverage gaps}, as limited testbenches cannot detect all bugs.
LLM-as-a-Judge (Finding 2) relies on \textbf{\reasonsignal{}s}, i.e., semantic correctness judgments derived from analyzing code logic. 
Their strength lies in \textit{semantic coverage}: reasoning can identify errors beyond what tests exercise. 
However, they suffer from \textit{reasoning hallucination}, producing inconsistent judgments for \yg{execution-equivalent} code.

% These findings reveal two critical phenomena in existing reranking methods: execution-based methods (Finding 1) leverage \textbf{\execsignal{}s}, i.e., deterministic pass/fail outcomes where functionally equivalent code must produce identical results; LLM-as-a-Judge (Finding 2) relies on \textbf{\reasonsignal{}s}, i.e., semantic correctness judgments that capture intent beyond test coverage but suffer from hallucination.
% Crucially, the two signal types exhibit \textit{orthogonal error characteristics}: execution errors are \textit{systematic} (arising from coverage gaps), while reasoning errors are \textit{stochastic} (arising from hallucination).

% This orthogonality suggests that execution signals can anchor the randomness in reasoning, while reasoning signals compensate for limited test coverage.
% We formalize this intuition through a \dualchannel information model (Section~\ref{sec:theory}), where the \execchannel generates \testbench{}s solely from requirement $x$, and the \reasonchannel judges correctness from $(x, \hat{y})$ pairs without observing execution results.
% Our theoretical analysis proves that such \textit{acquisition independence} strictly outperforms interactive approaches where the reasoner observes execution feedback (Theorem~\ref{thm:independent}), leading to our key insight: \textbf{\execsignal{}s and \reasonsignal{}s provide complementary evidence that should be acquired independently and fused at the decision stage}.

\yg{Crucially, the two limitations differ in character: execution errors are systematic, missing the same behaviors whenever the tests fail to exercise them, whereas reasoning errors are stochastic, varying across implementations that behave identically.
This asymmetry is what makes the signals worth combining: execution outcomes can stabilize fluctuating judgments, while reasoning can cover what the tests leave untested.}

\yg{How to combine them is less obvious.
The intuitive option is interactive: let the reasoner observe execution results before judging.
Information theory bounds what this can gain, but only when the resulting verdict is a revision of the judgment the reasoner would have produced anyway (Section~\ref{sec:theory}); the bound does not rule out every interactive design.
Empirically, however, exposing execution feedback anchors the reasoner on test outcomes and costs accuracy (Section~\ref{sec:eval}).}
\yg{We therefore keep the channels apart, a design we formalize as a} \dualchannel information model (Section~\ref{sec:theory}): the \execchannel generates \testbench{}s solely from requirement $x$, \yg{the} \reasonchannel judges correctness from $(x, \hat{y})$ pairs without observing execution results\yg{, and the two signals meet only when the final choice is made.
Our ablations support the premise behind this design: reasoning does carry information that execution misses (Section~\ref{sec:eval}).}

% Para 5: Method Overview
\yg{We realize this design as} \textbf{\tool} (\textbf{E}xecution-\textbf{A}nchored \textbf{H}allucination \textbf{C}alibration), \yg{a reranking framework with one component per channel}.
\textbf{\tool-R} \yg{implements the \reasonchannel}: we curate a 47K dataset via multi-teacher distillation with \compilerinloop verification and fine-tune Qwen3-4B to acquire Verilog-specific judgment capabilities; at inference, majority voting across multiple samples yields \yg{aggregated} \reasonsignal{}s.
\textbf{\tool-T} \yg{implements the \execchannel}: we construct a 53K dataset using the same distillation pipeline and employ RAG techniques to produce high-quality \execsignal{}s.
\yg{The two signals meet only in \hierselect, the decision stage anticipated above:}
candidates sharing identical execution vectors are grouped into \yg{execution-equivalence clusters} and assigned uniform reasoning scores, then ranked by fusing execution pass rates with \yg{cluster-level} reasoning confidence.
\yg{Since equivalent implementations share a single score, the winning cluster no longer depends on which candidate the judge happens to favor.
This is how execution anchors reasoning, and what we mean by \textbf{calibration}: consistency of judgments, not conversion of scores into probabilities.}

% Para 6: Results Highlight
We evaluate \tool on VerilogEval-v2 and ResBench across nine code generation models, including commercial (GPT-5, DeepSeek-V3, and GLM-4.6), general-purpose (Qwen2.5-Coder, Open-Coder, and Seed-Coder), and Verilog-specialized (HaVen, VeriPrefer, and CodeV-R1) models.
Results demonstrate that:
(1) \tool improves average Pass@1 from 53.99\% to \textbf{65.10\%} on VerilogEval-v2 and from 53.18\% to \textbf{68.25\%} on ResBench, outperforming the strongest baseline by \textbf{+5.91\%} and \textbf{+4.96\%} respectively\yg{, and ranking first in 15 of 18 configurations};
(2) ablation studies confirm both components are essential, with \tool-R contributing +7.12\% and \tool-T contributing +5.75\% on average;
(3) \yg{the fusion weight $\alpha{=}0.6$ performs best in our sweep}, reflecting the higher reliability of execution signals;
(4) \tool is orthogonal to training-based approaches, providing gains of 8.93\% to 19.64\% across base, SFT, and RL-optimized models; and
(5) independent fusion outperforms interactive fusion \yg{across all evaluated configurations}.

% Para 7: Contributions - 精简为3个
In summary, this paper makes the following contributions:
\begin{itemize}[leftmargin=*]
\item \textbf{Problem Formulation and Empirical Findings.} 
We formalize the Verilog \codererank problem and conduct the first systematic empirical study across nine code generation models and two benchmarks. 
% We identify two critical limitations: poor domain transferability of existing methods and \reasonhalluc in LLM-as-a-Judge, providing actionable insights for the research community.

\item \textbf{Dual-Channel Framework and Domain Resources.} 
We propose {\tool}, a \dualchannel reranking framework that \yg{acquires execution and reasoning signals independently and fuses them only at the decision stage, a choice supported by our information-theoretic analysis (Theorem~\ref{thm:independent}) and by a direct comparison against interactive fusion.
It also yields two reusable resources, VeriJudge-47K and VeriTest-53K, curated by multi-teacher distillation with \compilerinloop verification, along with the 4B judge and \testbench generator trained on them}. 

\item \textbf{State-of-the-Art Performance.} 
{\tool} \yg{attains the best average Pass@1 on both benchmarks and the best reranking accuracy in 15 of the 18 configurations} (9 models $\times$ 2 benchmarks), recovering over 60\% of the gap between Pass@1 and Pass@10. 
\end{itemize}

To support future research, we open-source our trained models and code\footnote{\url{https://github.com/NTDXYG/EAHC_CODE}} to facilitate reproducibility and future research in Verilog code generation.

% \noindent\textbf{Organization.}
% The remainder of this paper is organized as follows.
% Section~\ref{sec:background} introduces the background of Verilog code generation and reviews related work.
% Section~\ref{sec:empirical} presents our systematic empirical study that motivates the proposed approach.
% Section~\ref{sec:method} \yg{presents an information-theoretic design rationale and details} the {\tool} framework.
% Section~\ref{sec:eval} reports the experimental setup and evaluation results.
% Section~\ref{sec:discuss} discusses the broader implications, design choices, and threats to validity.
% Finally, Section~\ref{sec:conclusion} concludes the paper.
\section{Background and Related Work}
\label{sec:background}

\subsection{Verilog Code Generation}

Verilog is a hardware description language widely used for designing digital circuits, including ASICs and FPGAs~\cite{flake2020verilog}.
Recent work has explored leveraging LLMs for automated Verilog generation.
Early efforts fine-tuned code models on curated Verilog corpora~\cite{thakur2024verigen, liu2024rtlcoder}, while more recent approaches incorporate reinforcement learning with compiler feedback~\cite{wang2025large, wang2025verireason} or retrieval-augmented generation~\cite{ahmed2022verifix, qi2025verirag}.
Despite these advances, Verilog generation remains challenging due to the scarcity of training data and the complexity of hardware semantics.

\paragraph{Task Formulation.}
Given a natural language specification $x$ describing the desired hardware functionality, the goal of Verilog code generation is to produce a syntactically correct and functionally accurate Verilog module $\hat{y}$:
\begin{equation}
    \hat{y} = \mathcal{M}(x; \theta)
\end{equation}
where $\mathcal{M}$ denotes an LLM with parameters $\theta$. 
Correctness is evaluated by executing $\hat{y}$ against a ground-truth \testbench $\mathcal{T}^*$.
% For a single candidate, the binary outcome is:
% \begin{equation}
%     \text{pass}(\hat{y}) = \mathbb{1}\left[\texttt{Simulate}(\hat{y}, \mathcal{T}^*) = \texttt{PASS}\right]
% \end{equation}
When sampling $n$ candidates, the standard metric Pass@$k$ ($k \leq n$) estimates the probability that at least one of $k$ selected candidates is correct~\cite{chen2021evaluatinglargelanguagemodels}:
\begin{equation}
    \text{Pass@}k = \mathbb{E}_{\text{problems}}\left[1 - \frac{\binom{n-c}{k}}{\binom{n}{k}}\right]
\end{equation}
where $c$ denotes the number of correct candidates among $n$ samples. 
% This unbiased estimator avoids the high variance of naive sampling.
In practice, Pass@1 reflects single-attempt accuracy, while Pass@$k$ (e.g., $k{=}10$) reveals the upper-bound potential achievable through effective candidate selection.

\subsection{Code Reranking}
\label{subsec:reranking}

Code reranking addresses the problem of selecting the best candidate from multiple LLM-generated solutions.
Existing approaches can be broadly categorized into four paradigms:
(1) \textit{Generation probability}, which ranks candidates by their likelihood under the LLM~\cite{zhang2023av};
(2) \textit{Semantic matching}, which uses embedding similarity between requirements and code~\cite{neelakantan2022text};
(3) \textit{Execution verification}, which leverages test case execution as a selection signal~\cite{chen2022codet, shi2022natural};
and (4) \textit{LLM-as-a-Judge}, which employs LLMs to directly assess code correctness~\cite{yang2025code, zhao2025vrank}.
We provide detailed descriptions of baseline methods in Section~\ref{sec:empirical}.

\yg{The last two paradigms have been paired before, but only in coupled forms: a verifier consumes execution results as input features, or a model rewrites its code after seeing test failures.
Such coupling makes the judgment a function of the execution outcome, which is what our analysis bounds (Section~\ref{sec:theory}) and what our experiments find costly.
{\tool} instead keeps the two acquisitions apart, letting them meet only when the ranking is decided.}

\paragraph{Task Formulation.}
We formally define the \codererank problem as follows.
Given a natural language specification $x$ and a set of $k$ candidate implementations $\mathcal{Y}_k = \{\hat{y}_1, \hat{y}_2, \ldots, \hat{y}_k\}$ sampled from an LLM, the goal is to design a scoring function $R: \mathcal{X} \times \mathcal{Y} \rightarrow \mathbb{R}$ such that the top-ranked candidate maximizes correctness probability:
\begin{equation}
    \hat{y}^* = \arg\max_{\hat{y}_i \in \mathcal{Y}_k} R(x, \hat{y}_i)
\end{equation}
The objective is to maximize Pass@1 after reranking, thereby converting the latent potential of Pass@$k$ into realized single-attempt accuracy.
\begin{table*}[t]
\centering
\caption{Empirical study results on VerilogEval-v2 and ResBench. Best baseline results per column are \textbf{bolded}. ``--'' indicates unavailable probability for API-based models.}
\label{tab:empirical_results}
\footnotesize
\resizebox{\textwidth}{!}{
\begin{tabular}{l|ccccccccc|c}
\toprule
\textbf{Method} & \textbf{GPT-5} & \textbf{DS-V3} & \textbf{GLM-4} & \textbf{QC} & \textbf{OC} & \textbf{SC} & \textbf{HaVen} & \textbf{VeriPref} & \textbf{CodeV} & \textbf{Avg.} \\
\midrule
\multicolumn{11}{c}{\cellcolor{gray!15}\textit{VerilogEval-v2}} \\
\midrule
Pass@1 & 85.90 & 73.08 & 76.28 & 33.33 & 31.41 & 47.44 & 40.38 & 41.67 & 56.41 & 53.99 \\
\cdashline{1-11}
Probability & -- & -- & -- & 25.64 & 28.21 & 44.23 & 33.33 & 37.82 & 43.59 & 35.47 \\
CodeReviewer & -- & -- & -- & 25.64 & 28.85 & 43.59 & 32.69 & 37.82 & 43.59 & 35.36 \\
CodeRank-Q & \textbf{84.62} & 73.08 & 80.77 & 30.13 & 35.26 & 48.72 & 40.38 & 41.03 & 53.21 & 54.13 \\
CodeRank-J & 83.33 & 75.00 & 82.69 & 30.77 & 33.97 & 47.44 & 36.54 & 37.18 & 56.41 & 53.70 \\
CodeT-Self & 78.21 & 73.08 & 82.69 & 35.90 & 38.46 & 48.72 & 41.03 & 37.18 & 57.05 & 54.70 \\
CodeT-GPT & 78.21 & 72.44 & \textbf{85.90} & \textbf{39.10} & 39.10 & \textbf{52.56} & \textbf{49.36} & \textbf{53.21} & \textbf{62.82} & \textbf{59.19} \\
DiTing-1.5B & 82.69 & 74.36 & 82.05 & 35.26 & 36.54 & 48.08 & 42.31 & 46.79 & 55.13 & 55.91 \\
DiTing-7B & 83.33 & \textbf{76.28} & 82.69 & 36.54 & \textbf{39.74} & 51.28 & 46.79 & 50.64 & 58.97 & 58.47 \\
\cdashline{1-11}
Pass@10 (Oracle) & 92.31 & 85.26 & 92.95 & 53.21 & 56.41 & 67.31 & 58.33 & 66.67 & 69.23 & 71.30 \\
\midrule
\multicolumn{11}{c}{\cellcolor{gray!15}\textit{ResBench}} \\
\midrule
Pass@1 & 73.21 & 64.29 & 67.86 & 41.07 & 42.86 & 42.86 & 46.43 & 46.43 & 53.57 & 53.18 \\
\cdashline{1-11}
Probability & -- & -- & -- & 30.36 & 33.93 & 35.71 & 53.57 & 42.86 & 42.86 & 39.88 \\
CodeReviewer & -- & -- & -- & 32.14 & 30.36 & 33.93 & 57.14 & 41.07 & 44.64 & 39.88 \\
CodeRank-Q & \textbf{75.00} & 64.29 & 75.00 & 39.29 & 46.43 & 46.43 & 58.93 & 55.36 & 57.14 & 57.54 \\
CodeRank-J & 71.43 & 71.43 & 71.43 & 32.14 & 41.07 & 46.43 & 51.79 & 51.79 & 48.21 & 53.97 \\
CodeT-Self & 73.21 & 73.21 & 76.79 & 44.64 & 42.86 & 50.00 & 53.57 & 53.57 & \textbf{60.71} & 58.73 \\
CodeT-GPT & 73.21 & \textbf{76.79} & \textbf{80.36} & \textbf{57.14} & \textbf{53.57} & \textbf{60.71} & 55.36 & 53.57 & 58.93 & \textbf{63.29} \\
DiTing-1.5B & 69.64 & 67.86 & 67.86 & 48.21 & 46.43 & 55.36 & 60.71 & 50.00 & 58.93 & 58.33 \\
DiTing-7B & 71.43 & 64.29 & 71.43 & 48.21 & 48.21 & 46.43 & \textbf{66.07} & \textbf{58.93} & \textbf{60.71} & 59.52 \\
\cdashline{1-11}
Pass@10 (Oracle) & 85.71 & 83.93 & 87.50 & 73.21 & 71.43 & 73.21 & 75.00 & 75.00 & 75.00 & 77.78 \\
\bottomrule
\end{tabular}
}
\end{table*}

\section{Empirical Study}
\label{sec:empirical}

To understand the limitations of existing reranking methods on Verilog code generation, we conduct a systematic empirical study.
% This section presents our experimental setup, followed by two key findings that motivate our approach.

\subsection{Experimental Setup}

\noindent\textbf{Datasets.}
We evaluate on two established Verilog benchmarks:
(1) \textbf{VerilogEval-v2}~\cite{liu2023verilogeval}, containing 156 problems with human-written specifications and golden \testbench{}s;
(2) \textbf{ResBench}~\cite{guo2025resbench}, comprising 56 problems focused on more complex, realistic hardware designs.
Both benchmarks provide ground-truth \testbench{}s for functional verification.

\noindent\textbf{Code Generation Models.}
To ensure comprehensive coverage, we select 9 representative LLMs spanning three categories with diverse capabilities and specializations:
(1) \textit{Commercial models}: GPT-5~\cite{openai2025gpt5}, DeepSeek-V3.2~\cite{liu2024deepseek}, and GLM-4.6~\cite{glm46}, representing state-of-the-art proprietary systems with strong general reasoning abilities;
(2) \textit{General-purpose code models}: Qwen2.5-Coder-7B~\cite{hui2024qwen2}, OpenCoder-8B~\cite{huang2025opencoder}, and Seed-Coder-8B~\cite{seed2025seed}, which are open-source models optimized for code generation across multiple programming languages;
(3) \textit{Verilog-specialized models}: HaVen~\cite{yang2025haven}, VeriPrefer~\cite{wang2025insights}, and CodeV-R1~\cite{zhu2025qimeng}, which are fine-tuned specifically on Verilog corpora to enhance hardware code generation.
For each model, we sample $k{=}10$ candidates per problem using temperature $\tau{=}1.0$ to ensure sufficient diversity among generated solutions.

\noindent\textbf{Existing Reranking Methods.}
We evaluate representative methods from four paradigms:

(1) \textit{Generation Probability.}
\textbf{Prob} ranks candidates by their length-normalized log-probability under the generation model:
\begin{equation}
    R_{\text{prob}}(\hat{y}) = \frac{1}{|\hat{y}|} \sum_{t=1}^{|\hat{y}|} \log p_\theta(y_t \mid x, y_{<t})
\end{equation}
\textbf{CodeReviewer}~\cite{zhang2023av} extends this by combining forward generation probability with backward reconstruction likelihood:
\begin{equation}
    R_{\text{rev}}(x, \hat{y}) = \log p(\hat{y} \mid x) + \log p(x \mid \hat{y})
\end{equation}
The backward term measures how well the code can reconstruct the original requirement, providing a mutual information perspective.

(2) \textit{Semantic Matching.}
\textbf{CodeRank}~\cite{neelakantan2022text} projects requirements and code into a shared embedding space and ranks by cosine similarity:
\begin{equation}
    R_{\text{embed}}(x, \hat{y}) = \frac{\phi(x)^\top \phi(\hat{y})}{\|\phi(x)\| \cdot \|\phi(\hat{y})\|}
\end{equation}
where $\phi(\cdot)$ is a pretrained code embedding model. We evaluate two embedding models: Qwen3-Embedding~\cite{zhang2025qwen3} and Jina-Code-v2~\cite{kryvosheieva2025efficient}.

(3) \textit{Execution Verification.}
\textbf{CodeT}~\cite{chen2022codet} generates \testbench{}s alongside candidates and uses execution agreement as the ranking signal.
Given a generated \testbench set $\mathcal{T}$, candidates are clustered by execution outcomes, and cluster scores combine code count with test count:
\begin{equation}
    R_{\text{codet}}(\hat{y}) = |\mathcal{C}_{\hat{y}}| \times |\mathcal{T}_{\hat{y}}|
\end{equation}
where $\mathcal{C}_{\hat{y}}$ is the set of candidates sharing the same execution vector as $\hat{y}$, and $\mathcal{T}_{\hat{y}}$ is the set of tests they all pass.
We evaluate two variants: CodeT-Self (testbenches generated by the same model) and CodeT-GPT (testbenches generated by GPT-5).

(4) \textit{LLM-as-a-Judge.}
\textbf{Code-DiTing}~\cite{yang2025code} employs fine-tuned judge models to assess code correctness.
Given $n$ independent judgments, the score is computed via majority voting:
\begin{equation}
    R_{\text{judge}}(x, \hat{y}) = \frac{1}{n} \sum_{j=1}^{n} \mathbb{1}\left[F_\phi^{(j)}(x, \hat{y}) = \texttt{Yes}\right]
\end{equation}
where $F_\phi^{(j)}$ denotes the $j$-th sampled judgment from the LLM judge $F_\phi$.
We evaluate two model sizes: Code-DiTing-1.5B and Code-DiTing-7B, with majority voting.

\noindent\textbf{Evaluation Metrics.}
We report Pass@1 as the primary metric, measuring the accuracy of the top-ranked candidate.
The \textit{original Pass@1} refers to the correctness of the candidate generated via greedy decoding (temperature $\tau{=}0$), representing the default LLM output without any reranking.
For reranking methods, Pass@1 reflects the correctness of the selected candidate after ranking the $k$ sampled alternatives.
We also report Pass@$k$ (Oracle) as the upper bound, representing the probability that at least one correct solution exists among $k{=}10$ samples.
The gap between original Pass@1 and Oracle quantifies the potential improvement achievable through effective reranking.

\noindent\textbf{Implementation Details.}
All experiments are conducted on NVIDIA RTX4090 GPUs.
For LLM-as-a-Judge methods, we set $n{=}3$ sampling rounds in majority voting with temperature 0.6.
For CodeT, we generate 5 \testbench{}s per problem.
Verilog simulation is performed using Icarus Verilog.

\begin{figure*}
    \centering
    \includegraphics[width=0.9\textwidth]{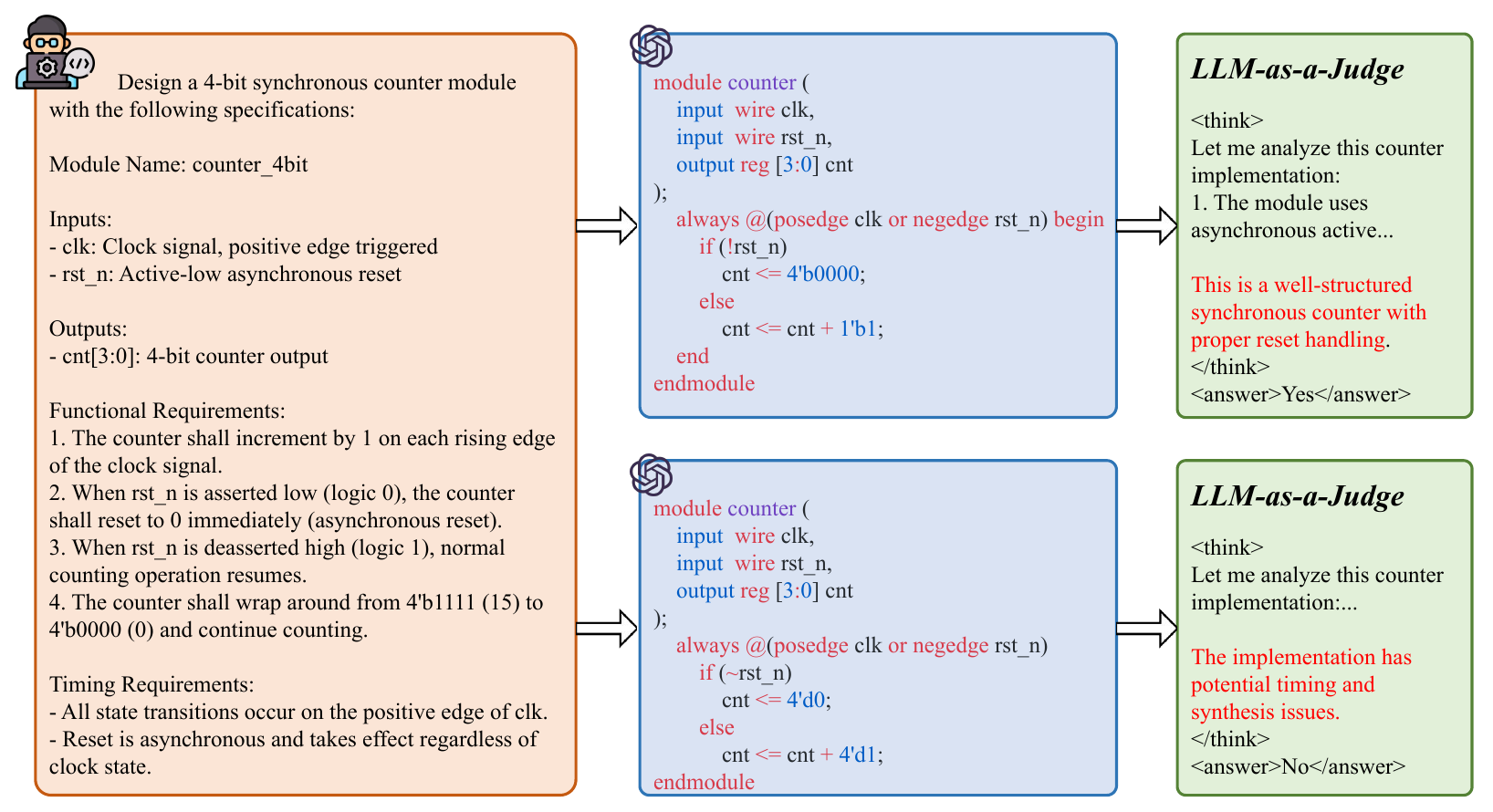}
    \caption{Example of reasoning hallucination}
    \label{fig:hallucination_example}
    \vspace{-0.5cm}
\end{figure*}

\subsection{Results and Analysis}

Table~\ref{tab:empirical_results} presents the reranking performance across all configurations.
We first analyze the characteristics of each method category, then summarize two critical findings.

\subsubsection{Method-wise Analysis}

\noindent\textbf{Generation Probability.}
Probability-based methods (Prob, CodeReviewer) consistently underperform, achieving only 35.47\% and 35.36\% average Pass@1 on VerilogEval-v2, even \textit{lower} than the original Pass@1 (53.99\%).
This degradation stems from distributional shift: Verilog syntax differs from the natural language and high-resource programming languages, causing unreliable likelihood estimates.
Notably, these methods are inapplicable to API-based models where token probabilities are unavailable.

\noindent\textbf{Semantic Matching.}
CodeRank variants show modest improvements on some models but remain inconsistent.
On VerilogEval-v2, CodeRank-Q achieves 54.13\% average Pass@1, only marginally above the original Pass@1.
The limitation lies in embedding models' inability to capture Verilog-specific semantics.

\noindent\textbf{Execution Verification.}
CodeT's performance is heavily dependent on the quality of generated \testbench{}s.
CodeT-Self, which relies on the same model to generate \testbench{}s, achieves limited improvements (54.70\% on VerilogEval-v2, 58.73\% on ResBench) due to the poor quality of self-generated test cases.
CodeT-GPT leverages GPT-5 for \testbench generation and achieves substantially better results (59.19\% and 63.29\%), emerging as the best-performing baseline on average.
However, this variant requires expensive API calls for each problem, limiting its practical applicability.
% Moreover, even GPT-5-generated \testbench{}s suffer from limited test coverage and occasional compilation failures, indicating that \testbench quality remains a fundamental bottleneck for execution-based methods.

\noindent\textbf{LLM-as-a-Judge.}
Code-DiTing shows competitive performance, with DiTing-7B achieving 58.47\% on VerilogEval-v2 and 59.52\% on ResBench.
Scaling from 1.5B to 7B yields consistent improvements (+2.5\% on average), suggesting reasoning capability matters.
However, despite being the most \textit{stable} baseline across models, LLM-as-a-Judge exhibits a critical flaw that limits its reliability (detailed in Finding 2).

\subsubsection{Finding 1: Poor Domain Transferability}
\label{subsec:finding1}

Reranking methods designed for general-purpose languages exhibit limited effectiveness on Verilog, with most failing to substantially outperform original Pass@1.
As shown in Table~\ref{tab:empirical_results}, only CodeT-GPT and DiTing-7B achieve meaningful improvements over the original Pass@1.
On VerilogEval-v2, the best baseline (CodeT-GPT, 59.19\%) still leaves a 12.11 percentage point gap to the oracle (71.30\%), indicating that \textbf{over 70\% of the potential improvement remains unrealized}.
Similar patterns emerge on ResBench, where CodeT-GPT achieves 63.29\% against an oracle of 77.78\%.

This poor transferability stems from the fundamental mismatch between method assumptions and HDL characteristics.
Probability-based methods assume reliable likelihood estimates, which fail under distributional shift.
Semantic matching relies on embeddings that lack HDL-specific representations.
Execution-based methods depend on test quality, yet self-generated \testbench{}s for Verilog exhibit low coverage and frequent compilation failures.

% \textit{Implication:} Effective Verilog reranking requires domain-specific components---either HDL-aware reasoning models or high-quality \testbench generation---rather than direct transfer from general-purpose methods.

\subsubsection{Finding 2: Reasoning Hallucination}
\label{subsec:finding2}

Among all baselines, LLM-as-a-Judge achieves the most consistent performance.
However, we identify a critical limitation: \textit{\reasonhalluc}, i.e., the tendency to produce inconsistent judgments for \yg{execution-equivalent} code.
\yg{To quantify this phenomenon, we call two candidates $\hat{y}_i$ and $\hat{y}_j$ execution-equivalent when they produce identical outputs on all ground-truth test cases, i.e., $\execvec_i = \execvec_j$.}
% We use this observational notion throughout in place of ``functional equivalence'' since a finite test suite cannot certify semantic or gate-level equivalence.}
We then examine whether LLM judgments respect this equivalence relationship.

\yg{They frequently do not.}
% For candidate pairs within the same equivalence class, DiTing-7B produces contradictory judgments (one approved, one rejected) in XX.X\% of cases on VerilogEval-v2.
Figure~\ref{fig:hallucination_example} illustrates a concrete example: two counter implementations with identical waveforms receive opposite verdicts.
One is praised for ``correct synchronous design,'' while the other is criticized for ``potential timing issues,'' despite both passing all functional tests.
The root cause is that LLMs reason at the \textit{token} level without grounding in \textit{execution semantics}.
Judgments become sensitive to superficial code variations rather than actual functional behavior.
While majority voting~\cite{chen2024more} reduces random noise, it cannot eliminate this systematic inconsistency.

% \textit{Implication:} Pure reasoning signals are insufficient for reliable Verilog reranking.
% However, execution results are \textbf{deterministic} for equivalent code---the same inputs always yield the same outputs.
% This determinism can serve as an \textit{anchor} to calibrate reasoning hallucination, ensuring that functionally equivalent implementations receive consistent scores.

% \subsection{Summary}

% Our empirical study reveals two key limitations of existing methods:
% (1) poor domain transferability from general-purpose code to Verilog;
% (2) \reasonhalluc in LLM-as-a-Judge leading to inconsistent judgments.
% These findings motivate our \tool framework, which combines execution anchoring with reasoning signals to address both limitations.
\section{Method}
\label{sec:method}
To address the limitations identified in Section~\ref{sec:empirical}, we propose \textbf{\tool}, a reranking framework that leverages execution signals to calibrate LLM reasoning, which is shown in Fig.~\ref{fig:method}.
% We first present the theoretical foundation (Section~\ref{sec:theory}), then describe {\tool}-R (Section~\ref{subsec:eahc-r}), {\tool}-T (Section~\ref{subsec:eahc-t}), and the \hierselect mechanism (Section~\ref{subsec:selection}).

\subsection{Framework Overview}
\label{subsec:overview}
Given a set of $k$ candidate implementations $\mathcal{Y}_k = \{\hat{y}_1, \ldots, \hat{y}_k\}$ for requirement $x$, \tool selects the optimal candidate through a two-stage process:
\begin{equation}
    \hat{y}^* = \underset{\hat{y} \in \mathcal{C}^*}{\arg\max}\, F_\phi(x, \hat{y}), \quad \text{where} \quad \mathcal{C}^* = \underset{\mathcal{C}_{\execvec}}{\arg\max}\, \hybridscore(\mathcal{C}_{\execvec})
\end{equation}
where $\mathcal{C}_{\execvec}$ denotes a \funcequivcluster based on execution vector $\execvec$, $\hybridscore$ is the fusion score combining execution and reasoning signals, and $F_\phi$ is the reasoning discriminator.

Algorithm~\ref{alg:eahc} summarizes the overall workflow:
(1) \textbf{Execution Anchoring}: generate \testbench{}s and obtain execution vectors for all candidates;
(2) \textbf{Equivalence Clustering}: group candidates by execution vectors;
(3) \textbf{Fusion Scoring}: compute hybrid scores for each cluster;
(4) \textbf{Hierarchical Selection}: select the best cluster, then the best individual within it.

\begin{algorithm}[t]
\caption{\tool Framework}
\label{alg:eahc}
\KwIn{Candidates $\mathcal{Y}_k$, requirement $x$, fusion weight $\alpha$}
\KwOut{Selected candidate $\hat{y}^*$}
\tcp{Phase 1: Execution Anchoring}
$\mathcal{T} \leftarrow \text{\tool-T}(x)$ \tcp*{Generate \testbench}
\For{each $\hat{y}_i \in \mathcal{Y}_k$}{
    $\execvec_i \leftarrow \texttt{Execute}(\hat{y}_i, \mathcal{T})$ \tcp*{Get execution vector}
}
\tcp{Phase 2: Equivalence Clustering}
$\{\mathcal{C}_{\execvec}\} \leftarrow \texttt{Cluster}(\mathcal{Y}_k, \{\execvec_i\})$ \tcp*{Group by execution}
\tcp{Phase 3: Fusion Scoring}
\For{each cluster $\mathcal{C}_{\execvec}$}{
    $\execscore(\mathcal{C}_{\execvec}) \leftarrow \texttt{PassRate}(\execvec)$\;
    $\reasonscore(\mathcal{C}_{\execvec}) \leftarrow \max_{\hat{y} \in \mathcal{C}_{\execvec}} F_\phi(x, \hat{y})$\;
    $\hybridscore(\mathcal{C}_{\execvec}) \leftarrow \alpha \cdot \execscore + (1-\alpha) \cdot \reasonscore$\;
}
\tcp{Phase 4: Hierarchical Selection}
$\mathcal{C}^* \leftarrow \arg\max_{\mathcal{C}_{\execvec}} \hybridscore(\mathcal{C}_{\execvec})$\;
$\hat{y}^* \leftarrow \arg\max_{\hat{y} \in \mathcal{C}^*} F_\phi(x, \hat{y})$\;
\Return{$\hat{y}^*$}
\end{algorithm}

\begin{figure*}[h]
    \centering
    \includegraphics[width=1\textwidth]{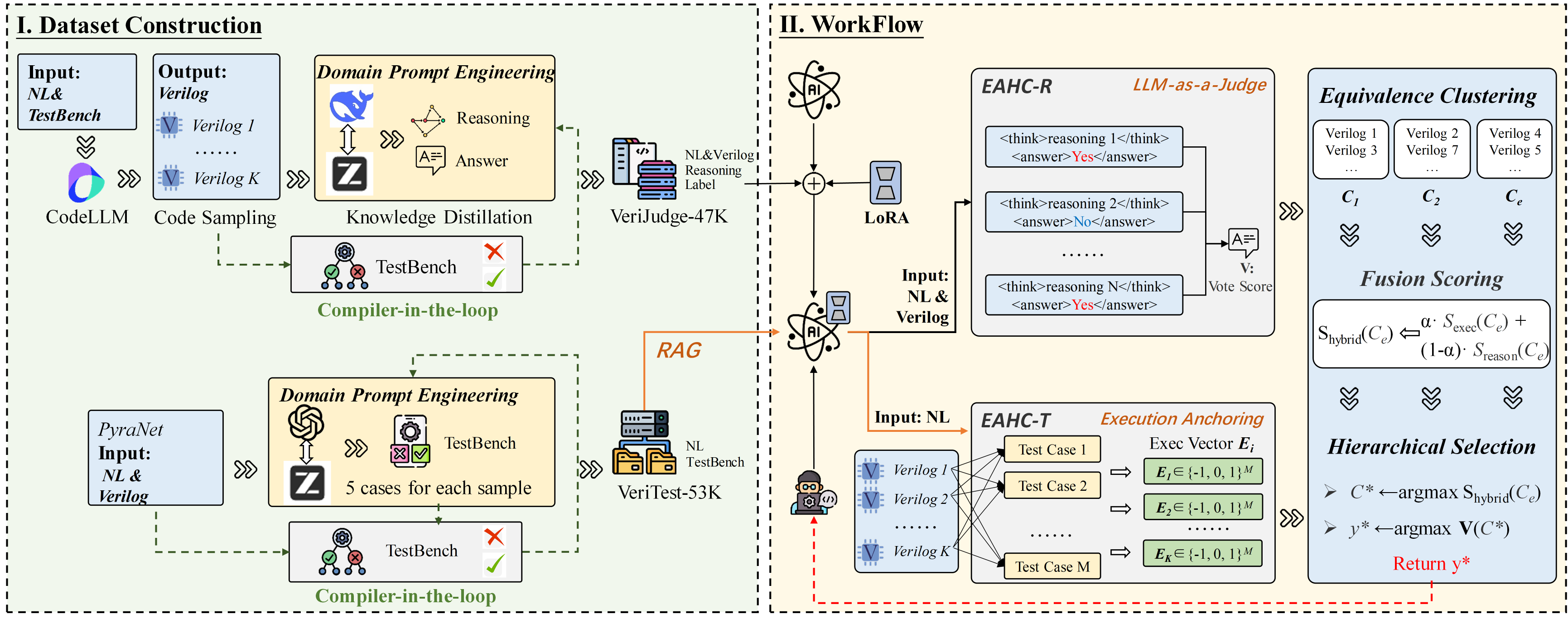}
    \caption{\tool Framework}
    \label{fig:method}
    \vspace{-0.5cm}
\end{figure*}

\subsection{\yg{Dual-Channel Information Model}}
\label{sec:theory}

\yg{This subsection asks two questions: what each channel adds to the other, and how the two should be combined.}

\subsubsection{Channel Independence}

We formalize correctness judgment as a \textit{dual-channel information acquisition problem}\yg{, in which each channel forms its signal without observing the other's}:

\begin{itemize}[leftmargin=*,nosep]
    \item \textbf{Execution Channel} $\rightarrow$ \textbf{Execution Signal} $\mathbf{E}$: Generates \testbench{}s solely from requirement $x$, without observing candidate code $\hat{y}$. The resulting execution outcomes (pass/fail on each test) constitute the \execsignal{}.
    \item \textbf{Reasoning Channel} $\rightarrow$ \textbf{Reasoning Signal} $R$: Judges correctness from $(x, \hat{y})$ pairs through semantic analysis, without observing execution results. The judgment (Yes/No with reasoning chain) constitutes the \reasonsignal{}.
\end{itemize}

\yg{Separating acquisition preserves the orthogonality of the two error sources: execution errors are \textit{systematic}, following from limited coverage, whereas reasoning errors are \textit{stochastic}, following from hallucination.}

\subsubsection{Signal Complementarity}
\yg{To reason about what the two signals jointly reveal, we adopt the approximation $P(\mathbf{E}, R \mid Y) \approx P(\mathbf{E} \mid Y) \cdot P(R \mid Y)$.
Because both channels read the same requirement $x$, this is a modeling choice rather than a property we can guarantee (Section~\ref{subsec:theory_bounds}).}

\begin{proposition}[Mutual Information Complementarity]
\label{thm:complementarity}
\yg{By the chain rule,} the joint use of $\mathbf{E}$ and $R$ satisfies
\begin{equation}
    I(Y; \mathbf{E}, R) = I(Y; \mathbf{E}) + I(Y; R \mid \mathbf{E}) \geq \max\{I(Y; \mathbf{E}), I(Y; R)\}
\end{equation}
with a strict gain iff $I(Y; R \mid \mathbf{E}) > 0$.
\end{proposition}

\yg{The identity holds unconditionally, and the approximation above is instead what licenses the additive fusion score below.
What neither settles is whether $I(Y; R \mid \mathbf{E}) > 0$, that is, whether reasoning stays informative once execution is known.
This is an empirical question, and our ablations answer it affirmatively (Section~\ref{sec:eval}).}

\subsubsection{Independent Fusion vs. Interactive Fusion}

A natural alternative is \textit{interactive fusion}, where the reasoner observes execution results before judging.
\yg{We analyze its \emph{post-hoc} form, in which an acquired $R$ is revised in light of $\mathbf{E}$; a judge that keeps $R$ intact while separately exploiting $\mathbf{E}$ lies outside the analysis.}

\begin{theorem}[\yg{Non-Superiority of Post-Hoc Interactive Fusion}]
\label{thm:independent}
Let $R$ be acquired independently of execution, and let $R' = h(R, \mathbf{E})$. Then
\begin{equation}
    I(Y; R' \mid \mathbf{E}) \leq I(Y; R \mid \mathbf{E}),
\end{equation}
\yg{with equality when $R'$ retains all $Y$-relevant information in $R$ given $\mathbf{E}$.}
\end{theorem}

\begin{proof}
\yg{Since $R'$ is computed from $(R, \mathbf{E})$, the variables form a Markov chain $Y \to (R, \mathbf{E}) \to R'$, and the Data Processing Inequality~\cite{beaudry2011intuitive} gives $I(Y; R', \mathbf{E}) \leq I(Y; R, \mathbf{E})$.
Subtracting $I(Y; \mathbf{E})$ from both sides yields the claim.}
\end{proof}

\yg{Equivalently, $I(Y; \mathbf{E}, R') \leq I(Y; \mathbf{E}, R)$: the pair an interactive design ends up with is never more informative than the pair we keep.}
The degenerate case is instructive: a reasoner that merely echoes execution, $R'{=}f(\mathbf{E})$, satisfies $I(Y; R' \mid \mathbf{E}){=}0$ and contributes nothing.
\tool-R therefore judges \textit{without} execution feedback.
\yg{We read this as a reason to acquire the signals independently, not as a proof that independence is optimal: the result bounds information content rather than the accuracy of any particular scoring rule, and {\tool} fuses the two signals linearly rather than optimally. Section~\ref{subsec:rq3} reports what the interactive design actually costs.}

\subsubsection{Posterior Fusion}

Assuming the conditional-independence approximation, Bayes' theorem yields the additive form
\begin{equation}
    \log P(Y{=}1 \mid \mathbf{E}, R) \propto \underbrace{\log P(\mathbf{E} \mid Y{=}1)}_{\execscore} + \underbrace{\log P(R \mid Y{=}1)}_{\reasonscore}
\end{equation}
with reliability weight $\alpha$:
\begin{equation}
    \hybridscore = \alpha \cdot \execscore + (1-\alpha) \cdot \reasonscore
\end{equation}
\yg{In practice we instantiate $\execscore$ and $\reasonscore$ as bounded surrogates for these log-likelihoods, a test pass rate and a judgment frequency (Sections~\ref{subsec:eahc-t} and~\ref{subsec:eahc-r}); the derivation thus motivates the additive form rather than yielding an exact posterior.}

%==============================================================================
\subsection{EAHC-R: Reasoning Discriminator}
\label{subsec:eahc-r}

As illustrated in Fig.~\ref{fig:method} (Part II, top branch), \tool-R functions as an LLM-as-a-Judge module that estimates the functional correctness probability $F_\phi(x, \hat{y}) \in [0, 1]$ for each candidate.

\subsubsection{Independence Constraint}
To preserve channel independence, \tool-R receives \textbf{only} the requirement $x$ and candidate code $\hat{y}$ as input, \textbf{excluding} any execution information (testbench, execution results, pass rates).
The reasoner judges correctness through pure semantic analysis of code logic.

\begin{figure}[t]
\begin{tcolorbox}[colback=gray!5, colframe=gray!50, title=\small\textbf{System Prompt for Code Verification}]
\small
\textit{Please serve as a Verilog code verification expert. Your task is to analyze the provided Verilog code and determine if it correctly implements the problem requirements.}

\textbf{Code Verification Task:}
\begin{enumerate}[nosep,leftmargin=*]
    \item \textbf{Code Analysis Phase}: Conduct a thorough logic and functional analysis of the provided Verilog code.
    \item \textbf{Test Case Generation Phase}: Generate comprehensive test cases covering expected functionality, edge cases, and potential failure modes.
    \item \textbf{Verification Phase}: Evaluate whether the code will correctly pass all generated test cases.
\end{enumerate}

\textbf{Output Format:} Provide only ``Yes'' or ``No'' within the \texttt{<answer></answer>} tags.
\end{tcolorbox}
\caption{Domain prompt template for reasoning generation in \tool-R.}
\vspace{-0.5cm}
\label{fig:prompt}
\end{figure}

\subsubsection{Training Data Construction}
As shown in Fig.~\ref{fig:method} (Part I, top pipeline), we construct the VeriJudge-47K high-quality $(x, \hat{y}, \text{reasoning}, \text{label})$ dataset through multi-teacher knowledge distillation:

\begin{enumerate}[leftmargin=*,nosep]
    \item \textbf{Candidate Generation}: We use Seed-Coder to sample diverse Verilog candidates from existing datasets with ground-truth testbenches~\cite{wei2025vericoder}.
    \item \textbf{Reasoning Generation}: For each candidate, we apply domain prompt engineering to elicit reasoning chains from two teacher models: DeepSeek-R1 and GLM-4. We deliberately employ multiple teachers to increase reasoning diversity and reduce single-model bias. The structured outputs contain \texttt{<think>}...\texttt{</think>} reasoning and \texttt{<answer>}Yes/No\texttt{</answer>} verdicts. Fig.~\ref{fig:prompt} illustrates our verification prompt template.
    \item \textbf{Compiler-in-the-loop Verification}: We execute candidates against ground-truth testbenches using Icarus Verilog. Only samples where the LLM's verdict aligns with execution results are retained, effectively filtering hallucinated reasoning.
\end{enumerate}

\subsubsection{Model Training and Inference}
We fine-tune Qwen3-4B on the curated VeriJudge-47K dataset using LoRA with PiSSA~\cite{meng2024pissa} initialization.
The lightweight 4B model enables single-GPU deployment while reducing high-latency inference cost compared to commercial API calls. 
At inference, we sample $n{=}3$ reasoning chains per candidate and aggregate via majority voting to obtain the \reasonsignal{}:
\begin{equation}
    F_\phi(x, \hat{y}) = \frac{1}{n} \sum_{j=1}^{n} \mathbb{1}\left[F_\phi^{(j)}(x, \hat{y}) = \texttt{Yes}\right]
\end{equation}
where $F_\phi^{(j)}$ denotes the $j$-th sampled reasoning chain.
\yg{$F_\phi$ is thus an endorsement frequency rather than a probability of correctness, and reranking depends only on the order it induces. Voting removes variance across samples for a single candidate; consistency across candidates is enforced later by execution anchoring (Section~\ref{subsec:selection}).}

%==============================================================================
\subsection{EAHC-T: Testbench Generator}
\label{subsec:eahc-t}

\tool-T generates high-quality testbenches to obtain execution signals for anchoring (see Fig.~\ref{fig:method}, bottom-left for data construction and right side for workflow).

\subsubsection{Independence Constraint}
As illustrated in the workflow (Fig.~\ref{fig:method}, right), \tool-T receives \textbf{only} the natural language requirement as input, \textbf{excluding} candidate implementations.
This design ensures \yg{that the same testbench fairly evaluates all $k$ candidates and that the execution signal $\mathbf{E}$ is acquired without access to the reasoning signal $R$}.

\subsubsection{VeriTest-53K: Testbench Knowledge Distillation}
As shown in Fig.~\ref{fig:method} (bottom-left), we construct the VeriTest-53K high-quality $(x, h, \text{testbenches})$ dataset through the following pipeline:

\begin{enumerate}[leftmargin=*,nosep]
    \item \textbf{Seed Data Collection}: We collect NL-Verilog pairs from PyraNet~\cite{nadimi2025pyranet} as seed examples, deliberately using a different data source from VeriJudge-47K. \yg{This separation keeps the two channels from being adapted on the same corpus, which limits one avenue for correlated behavior without by itself establishing conditional independence.}
    \item \textbf{Domain Prompt Engineering}: Using GPT-4o and GLM-4.6 as teachers, we generate 5 testbench cases for each sample through carefully designed prompts. Employing multiple teachers improves test coverage diversity, as different LLMs tend to generate test cases focusing on different aspects.
    \item \textbf{Compiler-in-the-loop Verification}: Generated testbenches are validated using Icarus Verilog. Only samples that compile successfully and follow correct format are retained.
\end{enumerate}

\yg{\noindent Both corpora are screened against the evaluation benchmarks; Section~\ref{subsec:contamination} quantifies the residual overlap.}

\begin{figure}[t]
\begin{tcolorbox}[colback=gray!5, colframe=gray!50, title=\small\textbf{RAG Prompt for Testbench Generation}]
\small
\textbf{System:} You are a Verilog testbench generation expert.

\textbf{User:} Please search the knowledge base for relevant testbenches and then generate the testbench.

Problem: \textit{[requirement $x$]}\\
Module Header: \textit{[interface $h$]}

\textbf{Assistant:} Here are the retrieved testbenches from the knowledge base:

\textit{[Top-$K$ retrieved examples $\mathcal{R}_K$]}

\textbf{User:} Please refer to the above knowledge base and generate \textit{five} Verilog testbench cases. Do not implement the module, only generate the testbench.

\textit{[One-shot example with expected format]}

\textbf{Assistant:} ... 
\end{tcolorbox}
\caption{RAG prompt template for testbench generation in \tool-T.}
\vspace{-0.5cm}
\label{fig:rag_prompt}
\end{figure}

\subsubsection{RAG-Enhanced Generation}
At inference, we employ a two-stage retrieval-augmented generation strategy, as illustrated in Fig.~\ref{fig:rag_prompt}.

\textbf{Stage 1: Similarity Retrieval.}
Given requirement $x$ and module interface $h$, we use BM25 to retrieve top-$K$ similar examples from VeriTest-53K:
\begin{equation}
    \mathcal{R}_K = \underset{(q,t) \in \mathcal{K}}{\text{Top-}K}\left[\text{BM25}(x \oplus h, q)\right]
\end{equation}

\textbf{Stage 2: Context-Augmented Generation.}
Inject retrieved examples as in-context demonstrations:
\begin{equation}
    \mathcal{T} = \mathcal{G}_T(x, h \mid \mathcal{R}_K)
\end{equation}
where $\mathcal{G}_T$ reuses the \tool-R tuned 4B model (Section~\ref{subsec:eahc-r}) without additional fine-tuning, relying on RAG with $K{=}5$ retrieved examples for domain adaptation.

\subsubsection{Execution Anchoring}
As shown in Fig.~\ref{fig:method} (right, EAHC-T block), each candidate $\hat{y}_i$ is executed against the generated testbench $\mathcal{T}$ containing $m$ test cases.
The execution vector $\execvec_i \in \{-1, 0, 1\}^m$ records the outcome:
\begin{equation}
    e_{ij} = \begin{cases} 
        1 & \text{if } \hat{y}_i \text{ passes test } t_j \\ 
        0 & \text{if } \hat{y}_i \text{ fails test } t_j \\ 
        -1 & \text{if compilation fails}
    \end{cases}
\end{equation}

\yg{Candidates with identical execution vectors ($\execvec_i = \execvec_j$) form an \textit{execution-equivalence} cluster; since $\mathcal{T}$ is generated rather than exhaustive, a cluster may still mix implementations that differ only where the tests are silent, which is exactly where the reasoning channel is needed.}
The execution score for a cluster $\mathcal{C}_{\execvec}$ is computed as the test pass rate:
\begin{equation}
    \execscore(\mathcal{C}_{\execvec}) = \frac{\sum_{j=1}^{m} \mathbb{1}[e_j = 1]}{\sum_{j=1}^{m} \mathbb{1}[e_j \neq -1]}
\end{equation}

%==============================================================================
\subsection{Hierarchical Selection}
\label{subsec:selection}

The final stage combines execution and reasoning signals through hierarchical selection (see Fig.~\ref{fig:method}, right side).

\subsubsection{Equivalence Clustering}
Candidates are grouped into \yg{execution-equivalence} clusters based on their execution vectors:
\begin{equation}
    \mathcal{C}_{\execvec} = \{\hat{y}_i \in \mathcal{Y}_k \mid \execvec_i = \execvec\}
\end{equation}

As illustrated in Fig.~\ref{fig:method}, candidates with identical execution behavior (e.g., Verilog 1, 3 in cluster $\mathcal{C}_1$; Verilog 2, 7 in cluster $\mathcal{C}_2$) are grouped together.

\subsubsection{Fusion Scoring}
For each cluster $\mathcal{C}_{\execvec}$, we compute the hybrid score by fusing execution and reasoning signals:
\begin{equation}
    \hybridscore(\mathcal{C}_{\execvec}) = \alpha \cdot \execscore(\mathcal{C}_{\execvec}) + (1-\alpha) \cdot \reasonscore(\mathcal{C}_{\execvec})
\end{equation}

The reasoning score aggregates individual judgments using $\max$:
\begin{equation}
    \reasonscore(\mathcal{C}_{\execvec}) = \max_{\hat{y}_i \in \mathcal{C}_{\execvec}} F_\phi(x, \hat{y}_i)
\end{equation}
which represents the most optimistic reasoning estimate within the cluster.
\yg{Taking the maximum lets one confident judgment carry a cluster, which is in tension with the consistency that anchoring aims for.}
We set $\alpha = 0.6$ by default, giving slightly higher weight to the more reliable execution signal.

\subsubsection{Two-Level Selection}

\textbf{Level 1: Cluster Selection.}
Select the cluster with highest hybrid score:
\begin{equation}
    \mathcal{C}^* = \arg\max_{\mathcal{C}_{\execvec}} \hybridscore(\mathcal{C}_{\execvec})
\end{equation}
For tie-breaking, we apply the priority order: $\hybridscore > \execscore > \reasonscore > |\mathcal{C}|$.

\textbf{Level 2: Individual Selection.}
Within the optimal cluster $\mathcal{C}^*$, select the candidate with highest reasoning score:
\begin{equation}
    \hat{y}^* = \arg\max_{\hat{y}_i \in \mathcal{C}^*} F_\phi(x, \hat{y}_i)
\end{equation}

% \textbf{Design Rationale.}
% This two-level design provides complementary guarantees:
% \begin{itemize}[leftmargin=*,nosep]
%     \item \textit{Cluster-level}: The selected candidate belongs to a cluster with high test pass rate, ensuring consistent execution behavior.
%     \item \textit{Individual-level}: Among functionally equivalent candidates, the one with best semantic quality (as judged by reasoning) is chosen.
% \end{itemize}

The final output $\hat{y}^*$ thus balances execution reliability with reasoning-based quality assessment.
% 主实验表格
\begin{table*}[t]
\centering
\caption{RQ1: Effectiveness evaluation results on VerilogEval-v2 and ResBench. \yg{The best reranking result per column is \textbf{bolded}; Pass@1 and the Pass@10 oracle are listed for reference.}}
\label{tab:rq1_results}
\footnotesize
\resizebox{\textwidth}{!}{
\begin{tabular}{l|ccccccccc|c}
\toprule
\textbf{Method} & \textbf{GPT-5} & \textbf{DS-V3} & \textbf{GLM-4} & \textbf{QC} & \textbf{OC} & \textbf{SC} & \textbf{HaVen} & \textbf{VeriPref} & \textbf{CodeV} & \textbf{Avg.} \\
\midrule
\multicolumn{11}{c}{\cellcolor{gray!15}\textit{VerilogEval-v2}} \\
\midrule
Pass@1 & 85.90 & 73.08 & 76.28 & 33.33 & 31.41 & 47.44 & 40.38 & 41.67 & 56.41 & 53.99 \\
\cdashline{1-11}
CodeT-Self & 78.21 & 73.08 & 82.69 & 35.90 & 38.46 & 48.72 & 41.03 & 37.18 & 57.05 & 54.70 \\
CodeT-GPT & 78.21 & 72.44 & \textbf{85.90} & 39.10 & 39.10 & 52.56 & 49.36 & 53.21 & 62.82 & 59.19 \\
DiTing-1.5B & 82.69 & 74.36 & 82.05 & 35.26 & 36.54 & 48.08 & 42.31 & 46.79 & 55.13 & 55.91 \\
DiTing-7B & 83.33 & 76.28 & 82.69 & 36.54 & 39.74 & 51.28 & 46.79 & 50.64 & 58.97 & 58.47 \\
\cdashline{1-11}
EAHC-T ($k$=1) & 83.97 & 76.28 & 80.13 & 35.90 & 37.82 & 48.72 & 41.03 & 38.46 & 55.77 & 55.34 \\
EAHC-T ($k$=3) & 83.97 & 76.92 & 82.05 & 39.10 & 41.03 & 50.64 & 44.87 & 45.51 & 57.69 & 57.98 \\
EAHC-R ($n$=1) & 83.97 & 76.92 & 79.49 & 39.74 & 46.79 & 61.54 & 50.00 & 56.41 & 61.54 & 61.82 \\
EAHC-R ($n$=3) & 82.05 & 78.85 & 80.13 & 41.67 & 47.44 & 60.90 & 52.56 & 57.69 & 63.64 & 62.77 \\
\textbf{EAHC} & \textbf{83.97} & \textbf{81.41} & 81.41 & \textbf{46.15} & \textbf{50.00} & \textbf{62.18} & \textbf{53.85} & \textbf{60.26} & \textbf{66.67} & \textbf{65.10} \\
\cdashline{1-11}
Pass@10 (Oracle) & 92.31 & 85.26 & 92.95 & 53.21 & 56.41 & 67.31 & 58.33 & 66.67 & 69.23 & 71.30 \\
\midrule
\multicolumn{11}{c}{\cellcolor{gray!15}\textit{ResBench}} \\
\midrule
Pass@1 & 73.21 & 64.29 & 67.86 & 41.07 & 42.86 & 42.86 & 46.43 & 46.43 & 53.57 & 53.18 \\
\cdashline{1-11}
CodeT-Self & 73.21 & 73.21 & 76.79 & 44.64 & 42.86 & 50.00 & 53.57 & 53.57 & 60.71 & 58.73 \\
CodeT-GPT & 73.21 & \textbf{76.79} & \textbf{80.36} & 57.14 & 53.57 & 60.71 & 55.36 & 53.57 & 58.93 & 63.29 \\
DiTing-1.5B & 69.64 & 67.86 & 67.86 & 48.21 & 46.43 & 55.36 & 60.71 & 50.00 & 58.93 & 58.33 \\
DiTing-7B & 71.43 & 64.29 & 71.43 & 48.21 & 48.21 & 46.43 & 66.07 & 58.93 & 60.71 & 59.52 \\
\cdashline{1-11}
EAHC-T ($k$=1) & 73.21 & 67.86 & 67.86 & 46.43 & 46.43 & 46.43 & 58.93 & 53.57 & 57.14 & 57.54 \\
EAHC-T ($k$=3) & 75.00 & 66.07 & 71.43 & 58.93 & 55.36 & 58.93 & 64.29 & 57.14 & 55.36 & 62.50 \\
EAHC-R ($n$=1) & 71.43 & 73.21 & 69.64 & 55.36 & 55.36 & 57.14 & 64.29 & 57.14 & 60.71 & 62.70 \\
EAHC-R ($n$=3) & 71.43 & 66.07 & 76.79 & 51.79 & 53.57 & 55.36 & 69.64 & 60.71 & 62.50 & 63.10 \\
\textbf{EAHC} & \textbf{76.79} & 73.21 & 78.57 & \textbf{60.71} & \textbf{60.71} & \textbf{67.86} & \textbf{71.43} & \textbf{60.71} & \textbf{64.29} & \textbf{68.25} \\
\cdashline{1-11}
Pass@10 (Oracle) & 85.71 & 83.93 & 87.50 & 73.21 & 71.43 & 73.21 & 75.00 & 75.00 & 75.00 & 77.78 \\
\bottomrule
\end{tabular}
}
\end{table*}
	
\section{Evaluation}
\label{sec:eval}

We evaluate {\tool} through \yg{four} research questions:
\begin{itemize}[leftmargin=*]
    \item \textbf{RQ1 (Overall Effectiveness)}: How effective is {\tool} compared to existing code reranking methods on Verilog generation tasks?
    
    \textit{Motivation}: Our empirical study (Section~\ref{sec:empirical}) revealed that existing methods suffer from poor domain transferability and reasoning hallucination. We evaluate whether {\tool}'s dual-channel design effectively addresses these limitations.
    
    \item \textbf{RQ2 (Orthogonality)}: Is {\tool} orthogonal and complementary to training-based optimization methods (SFT and RL)?
    
    \textit{Motivation}: Training-based methods (SFT, RL) have been widely adopted to improve code generation models. Since these methods optimize the model's generation distribution while {\tool} optimizes candidate selection, we investigate whether the two optimization dimensions are orthogonal and can be combined for additional gains across different training stages.
    
    \item \textbf{RQ3 (Independent vs. Interactive Fusion)}: Does independent channel fusion outperform interactive fusion \yg{in practice}?
    
    \textit{Motivation}: \yg{Our analysis bounds what a reasoner gains from observing execution results, but leaves open how a concrete implementation behaves. We therefore compare independent fusion against the alternative where the reasoner sees execution feedback before judging.}

    \item \textbf{\yg{RQ4 (Validity of the Execution Anchor)}}: \yg{How reliable are the testbenches that {\tool} generates, and how does the framework behave when they carry no information?}

    \textit{Motivation}: \yg{Clustering and scoring both depend on self-generated testbenches, so the execution channel must be validated rather than assumed. We measure it against the hidden ground-truth testbenches and examine the problems where it fails to separate candidates.}
\end{itemize}

%==============================================================================
\subsection{RQ1: Overall Effectiveness}
\label{subsec:rq1}

\begin{figure*}[t]
	\centering
	\subfigure[Impact of $\alpha$ on VerilogEval-v2]{%
		\includegraphics[width=0.48\textwidth]{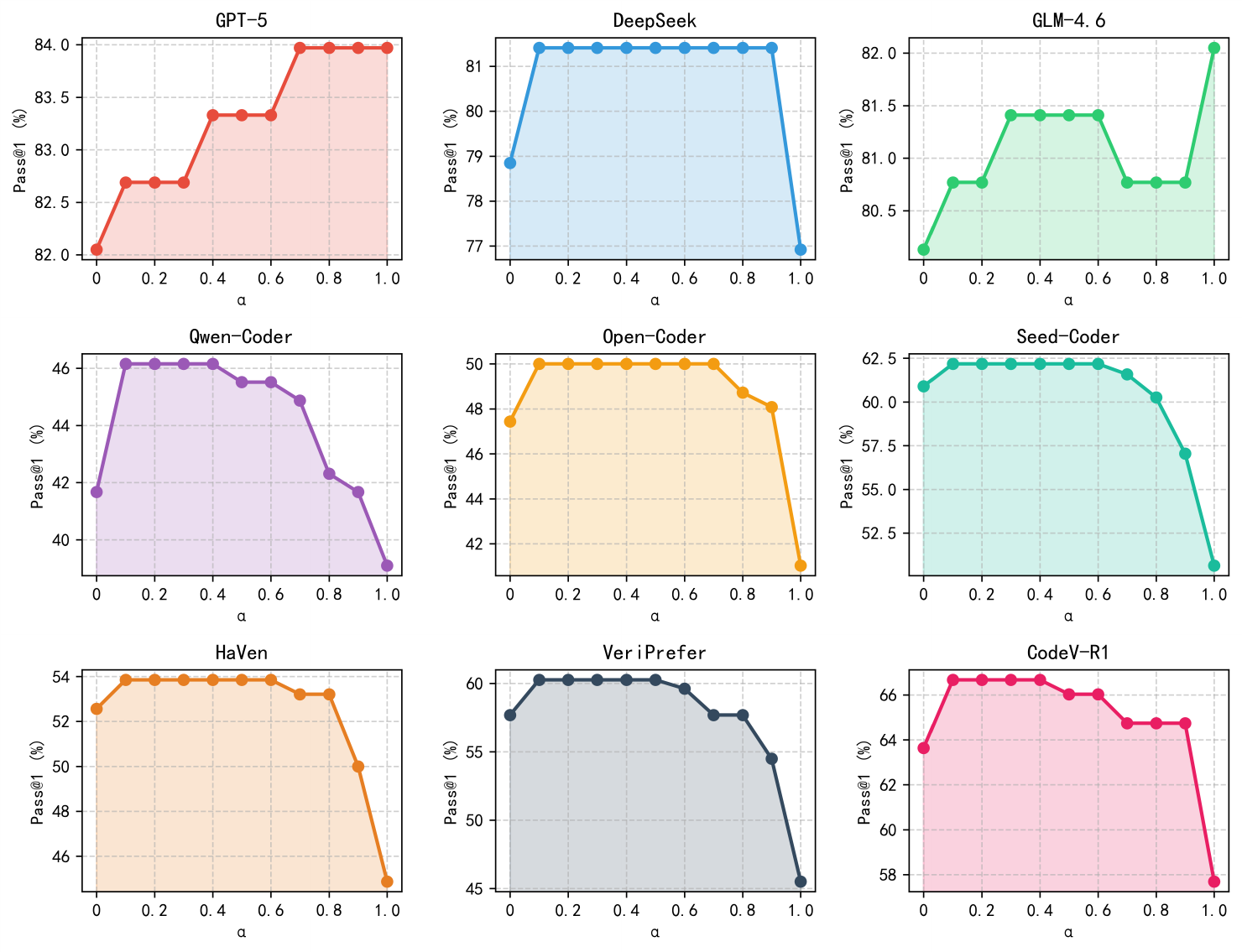}%
		\label{fig:alpha_verilog}%
	}%
	% \hfill
	\subfigure[Impact of $\alpha$ on ResBench]{%
		\includegraphics[width=0.48\textwidth]{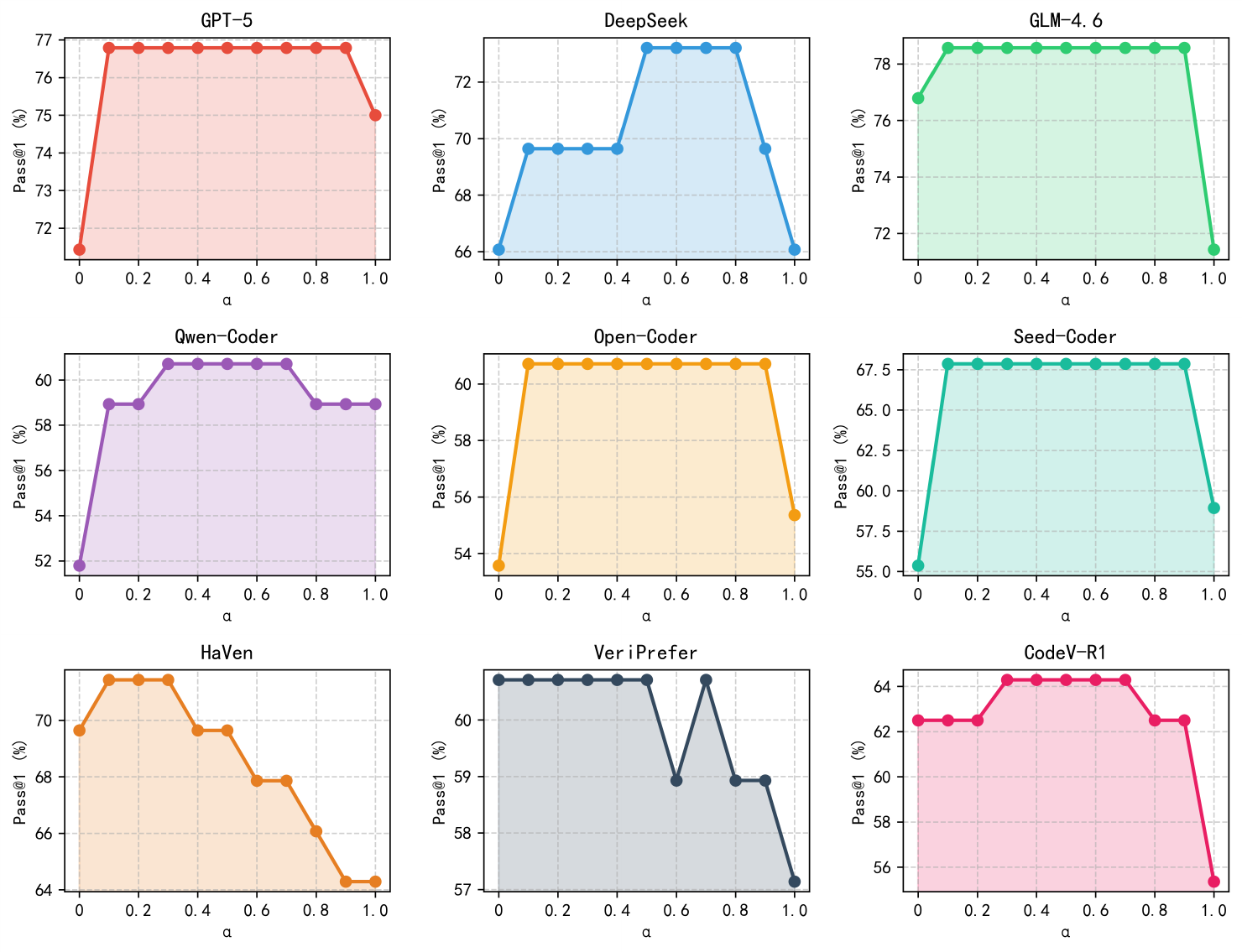}%
		\label{fig:alpha_resbench}%
	}%
	\caption{Sensitivity analysis of hyperparameter $\alpha$ across different LLMs.}
	\label{fig:alpha_analysis}
	\vspace{-0.5cm}
\end{figure*}
	
\subsubsection{Performance Comparison}

Table~\ref{tab:rq1_results} presents the experimental results on VerilogEval-v2 and ResBench.
\yg{{\tool} attains the highest average Pass@1 on both benchmarks and the best reranking accuracy in 15 of the 18 configurations.}

\textbf{Results on VerilogEval-v2.}
{\tool} achieves an average Pass@1 of \textbf{65.10\%}, outperforming the strongest baseline CodeT-GPT (59.19\%) by \textbf{+5.91\%} absolute improvement.
Compared to the original Pass@1 (53.99\%), {\tool} recovers \textbf{64.2\%} of the gap to the Pass@10 oracle (71.30\%).
The improvements are particularly pronounced on domain-specific models with weaker base performance: Open-Coder improves from 31.41\% to 50.00\% (\textbf{+18.59\%}), and Seed-Coder from 47.44\% to 62.18\% (\textbf{+14.74\%}).

\textbf{Results on ResBench.}
{\tool} achieves an average Pass@1 of \textbf{68.25\%}, surpassing CodeT-GPT (63.29\%) by \textbf{+4.96\%}.
This represents \textbf{61.3\%} recovery of the Pass@1-to-Pass@10 gap.
Notable improvements include Qwen-Coder (41.07\%$\rightarrow$60.71\%, \textbf{+19.64\%}) and HaVen (46.43\%$\rightarrow$71.43\%, \textbf{+25.00\%}), demonstrating \tool's effectiveness on challenging generation tasks.

\textbf{\yg{Analysis of Exceptions.}}
\yg{{\tool} is not the best reranker in three configurations, and CodeT-GPT wins all of them: GLM-4 on VerilogEval-v2 (81.41 vs.\ 85.90), and DS-V3 and GLM-4 on ResBench (73.21 vs.\ 76.79 and 78.57 vs.\ 80.36).
The three share a profile: a strong generator whose candidates mostly compile and pass, evaluated against CodeT-GPT's GPT-5-authored testbenches.
When almost every candidate is already correct, the ranking turns on the two error modes that remain.
The first lies in the execution channel, where a permissive testbench admits an incorrect implementation into a passing cluster: Section~\ref{subsec:rq4} reports false-accept rates of 12.8\% and 6.6\% and mixed-cluster rates of 23.8\% and 27.5\%.
On both quantities the GPT-5 testbenches are stronger on ResBench (5.9\% and 25.7\%), which accounts for the two exceptions there.
The second lies in the reasoning channel, where an incorrect candidate outranks every correct one on 5.7\% and 8.4\% of the problems that contain one (Table~\ref{tab:consistency}); anchoring reduces this drift without eliminating it, which explains the remaining exception on VerilogEval-v2.}

\textbf{Pairwise Significance Testing}
McNemar's test~\cite{mcnemar1947note} is appropriate for paired binary outcomes, comparing the number of problems where one method succeeds and the other fails.
Let $n_{01}$ denote problems where the baseline succeeds but {\tool} fails, and $n_{10}$ denote the reverse.
The test statistic is:
\begin{equation}
    \chi^2 = \frac{(|n_{01} - n_{10}| - 1)^2}{n_{01} + n_{10}}
\end{equation}

Compared to the best baseline CodeT-GPT, {\tool} shows highly significant improvements on both datasets. On VerilogEval-v2, the p-value across all models is $1.1\times10^{-10}$, and on ResBench, it is $1.8\times10^{-8}$. These results ($p<0.01$) indicate that \yg{the aggregate advantage over CodeT-GPT is unlikely to arise by chance, although, as noted above, CodeT-GPT remains ahead on three individual configurations}.

\subsubsection{Ablation Study}

We conduct ablation studies to understand the contribution of each component, as shown in Table~\ref{tab:rq1_results}.

\textbf{Execution Channel (\tool-T).}
Using only execution-based testcase scoring (\tool-T, $k$=3) achieves 57.98\% on VerilogEval-v2 and 62.50\% on ResBench.
The performance gap to full {\tool} (7.12\% and 5.75\%) confirms that reasoning signals provide substantial complementary value beyond execution feedback.

\textbf{Reasoning Channel (\tool-R).}
Using only reasoning-based hierarchical ranking (\tool-R, $n$=3) achieves 62.77\% on VerilogEval-v2 and 63.10\% on ResBench.
The performance degradation (2.33\% and 5.15\% vs. full \tool) demonstrates that execution anchoring provides crucial grounding for reasoning-based selection.

\textbf{Hierarchical Selection.}
Comparing \tool-R ($n$=1) with \tool-R ($n$=3), hierarchical selection improves performance from 61.82\% to 62.77\% on VerilogEval-v2 and from 62.70\% to 63.10\% on ResBench.
This validates our cluster-based design for mitigating reasoning hallucination through diversity preservation.

\textbf{Fusion Weight $\alpha$.}
Figure~\ref{fig:alpha_analysis} illustrates the sensitivity of {\tool} to the fusion weight $\alpha \in [0, 1]$.
Performance consistently peaks in the range $\alpha \in [0.1, 0.6]$, with optimal values around $\alpha = 0.3$--$0.6$.
Both extreme configurations ($\alpha=0$: reasoning only; $\alpha=1$: execution only) yield suboptimal results, with performance degrading sharply when $\alpha > 0.8$.
This confirms that the dual-channel fusion effectively leverages complementary information from both sources.

\begin{table}[htbp]
\centering
\caption{\yg{Consistency of {\tool}-R ($n$=3) on execution-equivalent candidates, aggregated over the nine generators. A pair is two candidates of one problem that pass all ground-truth tests; only pairs with fully parseable judgments are counted.}}
\label{tab:consistency}
\resizebox{\columnwidth}{!}{
\begin{tabular}{l|cc}
\toprule
\textbf{Measure} & \textbf{VerilogEval-v2} & \textbf{ResBench} \\
\midrule
Problems with $\geq$2 equivalent candidates (\%) & 65.2 & 71.6 \\
Equivalent pairs & 30,502 & 10,040 \\
\midrule
Opposite verdicts (\%) & 4.8 & 7.1 \\
Unequal scores (\%) & 15.3 & 27.4 \\
Identical code, unequal scores (\%) & 5.2 & 16.4 \\
Wrong candidate outranks all correct (\%) & 5.7 & 8.4 \\
\bottomrule
\end{tabular}
}
\end{table}

\subsubsection{\yg{The Inconsistency that Execution Anchoring Targets}}
\label{subsubsec:consistency}

\yg{\hierselect rests on the premise that {\tool}-R scores drift across candidates with identical behavior. We verify this on candidates of the same problem that pass every ground-truth test: they are execution-equivalent, so a consistent judge should score them identically. Table~\ref{tab:consistency} shows otherwise.}

\yg{Scores disagree in 15.3\% and 27.4\% of equivalent pairs, and since ranking depends on scores rather than binary verdicts, these are the operative rates. Majority voting ($n$=3) reduces within-candidate variance, lowering verdict flips from 7.5\% to 4.8\% and from 12.8\% to 7.1\%, but cross-candidate disagreement persists. Even byte-identical code receives different scores in 5.2\% and 16.4\% of pairs, where no semantic ambiguity can explain the gap. The cost is direct: on 5.7\% and 8.4\% of problems that contain a correct candidate, an incorrect one ranks highest, a loss voting cannot fix. Execution anchoring eliminates this drift by assigning one score per cluster, yielding +2.33\% and +5.15\% over {\tool}-R. Inconsistency also rises with generator strength (3.8\% on Seed-Coder vs.\ 7.5\% on GPT-5 and GLM-4), consistent with the narrowing margins in RQ1.}

\begin{tcolorbox}[colback=SeaGreen!10!CornflowerBlue!10,colframe=RoyalPurple!55!Aquamarine!100!,title=Summary of RQ1]
	{\tool} achieves the highest average Pass@1 on both benchmarks (65.10\% and 68.25\%, i.e.\ +5.91\% and +4.96\% over the best baseline) and ranks first in \yg{15 of 18 configurations, with CodeT-GPT still ahead on three strong-generator configurations}. Both execution and reasoning channels contribute positively, and \yg{their fusion performs best overall. Even after majority voting, reasoning scores disagree on 4.8\%--7.1\% of execution-equivalent pairs, the inconsistency execution anchoring absorbs}.
\end{tcolorbox}
%==============================================================================
\subsection{RQ2: Orthogonality}
\label{subsec:rq2}

\subsubsection{Motivation}
Training-based methods, including supervised fine-tuning (SFT) and reinforcement learning (RL), have been widely adopted to improve code generation models.
A natural question arises: \textit{Is {\tool} complementary to training-based approaches, or do they overlap in their improvements?}
We investigate whether {\tool} can provide consistent gains across different stages of the model training pipeline.

\subsubsection{Experimental Design}
We select the CodeV model family as our subject, which follows a typical three-stage training pipeline for domain-specific code generation:
\begin{itemize}[leftmargin=*,nosep]
    \item \textbf{Base Model (Qwen2.5-Coder)}: A general-purpose code LLM pre-trained on multi-language code corpora, serving as the foundation without any Verilog-specific optimization.
    \item \textbf{SFT Model (CodeV-R1)}: Built upon Qwen2.5-Coder, this model is supervised fine-tuned on curated Verilog code generation datasets to acquire domain-specific knowledge and coding patterns.
    \item \textbf{RL Model (CodeV-R1-RL)}: Starting from CodeV-R1, this model is further optimized using the DAPO algorithm~\cite{yu2025dapo} with execution feedback as reward signals.
\end{itemize}

This progression (Base $\rightarrow$ SFT $\rightarrow$ RL) represents a common paradigm in building domain-specific code generation models, where each stage incrementally improves the model's generation capability.
For each model variant, we evaluate Pass@1, Pass@1 with {\tool}-T, Pass@1 with {\tool}-R, Pass@1 with full {\tool}, and Pass@10 (Oracle upper bound) to analyze whether \tool's improvements are consistent across different training stages.

\begin{figure}[t]
	\centering
	\subfigure[Results on VerilogEval-v2]{%
		\includegraphics[width=0.48\textwidth]{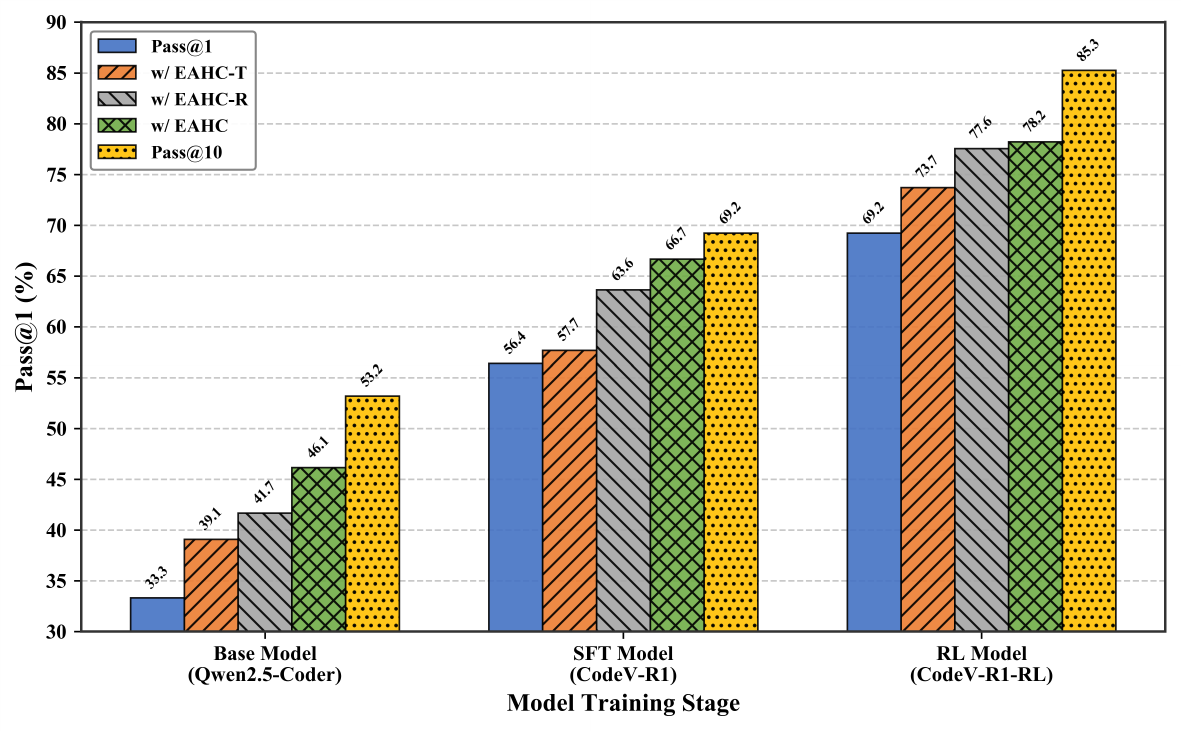}%
		\label{fig:rq2_verilog}%
	}%
	\hfill
	\subfigure[Results on ResBench]{%
		\includegraphics[width=0.48\textwidth]{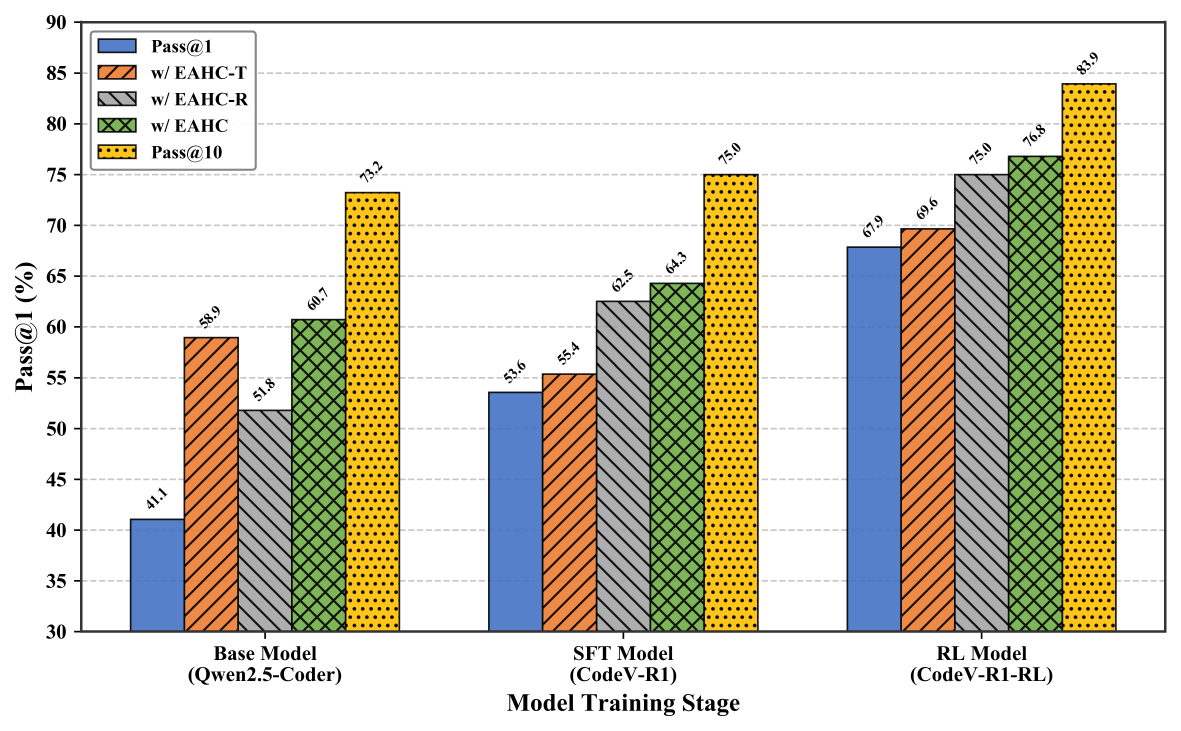}%
		\label{fig:rq2_resbench}%
	}%
	\caption{Orthogonality with Training-based Methods.}
	\label{fig:rq2_analysis}
	\vspace{-0.5cm}
\end{figure}

\subsubsection{Results Analysis}

Figure~\ref{fig:rq2_analysis} presents the orthogonality analysis across three training stages.

\noindent\textbf{Consistent Improvements Across All Training Stages.}
{\tool} delivers substantial improvements regardless of the underlying model's training stage.
On VerilogEval-v2, {\tool} improves Pass@1 by \textbf{+12.82\%} for Base (33.33\%$\rightarrow$46.15\%), \textbf{+10.26\%} for SFT (56.41\%$\rightarrow$66.67\%), and \textbf{+8.98\%} for RL (69.23\%$\rightarrow$78.21\%).
On ResBench, the corresponding improvements are \textbf{+19.64\%}, \textbf{+10.72\%}, and \textbf{+8.93\%}.
These consistent gains confirm that generation optimization (via training) and selection optimization (via {\tool}) operate on orthogonal dimensions.

\noindent\textbf{Complementary Effects.}
Training and reranking address fundamentally different aspects of code generation.
Training (SFT$\rightarrow$RL) progressively improves base Pass@1 from 33.33\% to 69.23\% on VerilogEval-v2 by enhancing candidate quality.
Meanwhile, {\tool} recovers a significant portion of the Pass@1-to-Pass@10 gap at each stage: 64.5\% (Base), 80.0\% (SFT), and 56.0\% (RL).
Notably, even the RL-optimized model, which already incorporates execution feedback during training, still benefits substantially from {\tool}'s selection strategy, indicating that generation-time optimization and inference-time selection capture complementary signals.

% \noindent\textbf{Component Contribution.}
% Both channels contribute positively across all training stages.
% Interestingly, \tool-R exhibits stronger improvements than \tool-T on SFT and RL models (e.g., +7.23\% vs. +1.28\% for SFT on VerilogEval-v2).
% This suggests that reasoning-based evaluation becomes increasingly effective as model quality improves, since higher-quality candidates exhibit more nuanced differences that benefit from semantic analysis.

% \noindent\textbf{Practical Implications.}
% These results establish \tool as a model-agnostic enhancement that can be readily deployed on any code generation model, including pre-trained, fine-tuned, and RL-optimized variants, without additional training overhead.
% The orthogonality between training and reranking enables practitioners to combine both strategies for maximum effectiveness.

\begin{tcolorbox}[colback=SeaGreen!10!CornflowerBlue!10,colframe=RoyalPurple!55!Aquamarine!100!,title=Summary of RQ2]
\yg{Within this model family, {\tool} adds +8.93\%--+19.64\% at every training stage, indicating that generation optimization and selection optimization address different aspects of code generation and can be combined. Whether the same holds for other training pipelines remains to be tested.}
\end{tcolorbox}

%==============================================================================
%==============================================================================
\subsection{RQ3: Independent vs. Interactive Fusion}
\label{subsec:rq3}

\subsubsection{Motivation}
An intuitive alternative to our design is \textit{interactive fusion}, where the reasoner observes execution results before judging.
\yg{Theorem~\ref{thm:independent} bounds what such interaction can gain when it amounts to revising an acquired judgment, but the bound admits equality and does not predict accuracy loss.
Whether interactive fusion actually hurts is therefore empirical;} we compare the two strategies across all nine models.

\subsubsection{Experimental Design}
We compare two fusion strategies under identical experimental settings:
\begin{itemize}[leftmargin=*,nosep]
    \item \textbf{Interactive Fusion}: \tool-R receives execution results (pass/fail status, error messages) as additional context before making judgments. The reasoning and execution signals are fused \textit{interactively} within the reasoner.
    \item \textbf{Independent Fusion}: \tool-R judges correctness without observing any execution information. Signals are fused \textit{independently} at the decision stage via $\hybridscore = \alpha \cdot \execscore + (1-\alpha) \cdot \reasonscore$.
\end{itemize}

\yg{Both use the same {\tool}-R checkpoint and differ only in the prompt; the comparison thus isolates prompt-level interaction rather than end-to-end training with execution context.}

% 融合策略对比表格
\begin{table*}[t]
\centering
\caption{Independent fusion vs. interactive fusion across all models. Best results per column are \textbf{bolded}.}
\label{tab:fusion_comparison}
\footnotesize
% \resizebox{\textwidth}{!}{
\begin{tabular}{l|ccccccccc|c}
\toprule
\textbf{Strategy} & \textbf{GPT-5} & \textbf{DeepSeek-V3} & \textbf{GLM-4} & \textbf{Qwen-Coder} & \textbf{Open-Coder} & \textbf{Seed-Coder} & \textbf{HaVen} & \textbf{VeriPrefer} & \textbf{CodeV} & \textbf{Avg.} \\
\midrule
\multicolumn{11}{c}{\cellcolor{gray!15}\textit{VerilogEval-v2}} \\
\midrule
Interactive Fusion & 82.69 & 78.21 & 79.49 & 41.67 & 44.87 & 56.41 & 48.08 & 52.56 & 61.54 & 60.61 \\
Independent Fusion & \textbf{83.97} & \textbf{81.41} & \textbf{81.41} & \textbf{46.15} & \textbf{50.00} & \textbf{62.18} & \textbf{53.85} & \textbf{60.26} & \textbf{66.67} & \textbf{65.10} \\
\cdashline{1-11}
$\Delta$ & +1.28 & +3.20 & +1.92 & +4.48 & +5.13 & +5.77 & +5.77 & +7.70 & +5.13 & \textbf{+4.49} \\
\midrule
\multicolumn{11}{c}{\cellcolor{gray!15}\textit{ResBench}} \\
\midrule
Interactive Fusion & 75.00 & 69.64 & 75.00 & 55.36 & 55.36 & 62.50 & 66.07 & 55.36 & 58.93 & 63.69 \\
Independent Fusion & \textbf{76.79} & \textbf{73.21} & \textbf{78.57} & \textbf{60.71} & \textbf{60.71} & \textbf{67.86} & \textbf{71.43} & \textbf{60.71} & \textbf{64.29} & \textbf{68.25} \\
\cdashline{1-11}
$\Delta$ & +1.79 & +3.57 & +3.57 & +5.35 & +5.35 & +5.36 & +5.36 & +5.35 & +5.36 & \textbf{+4.56} \\
\bottomrule
\end{tabular}
% }
\end{table*}

\subsubsection{Results}
Table~\ref{tab:fusion_comparison} presents the comparison.
The results \yg{support our design choice}: \textbf{independent fusion outperforms interactive fusion on every model and benchmark}, with average gains of \textbf{+4.49\%} on VerilogEval-v2 and \textbf{+4.56\%} on ResBench.

\textbf{Performance Degradation Pattern.}
\yg{The gap is smallest for the strongest models (GPT-5 +1.28, GLM-4 +1.92 on VerilogEval-v2) and largest for the weaker ones (up to +7.70 for VeriPrefer).
Interactive fusion thus appears most harmful when reasoning must discriminate among lower-quality candidates.}

\textbf{Anchoring Bias Analysis.}
Qualitative inspection of \tool-R's reasoning chains reveals pronounced \textit{anchoring bias} under interactive fusion:
\begin{itemize}[leftmargin=*,nosep]
    \item \textbf{Pass Anchoring}: When all tests pass, the reasoner tends to approve the candidate even when subtle logical errors exist that the testbench fails to cover.
    \item \textbf{Fail Anchoring}: When tests fail, the reasoner over-penalizes the candidate even when the failure stems from testbench limitations rather than code errors.
\end{itemize}

\yg{Anchoring is one instance of the post-hoc revision that Theorem~\ref{thm:independent} covers: the reasoning channel drifts toward echoing execution outcomes instead of contributing independently.
The accuracy loss is consistent with the lossy end of that bound rather than with equality; whether a judge trained to distrust unreliable feedback could approach equality remains open.}

\textbf{Comparison with Execution-Only.}
Notably, interactive fusion (60.61\% / 63.69\%) performs only marginally better than execution-only \tool-T (57.98\% / 62.50\%), suggesting that the reasoning signal's contribution is largely ``absorbed'' by execution information.
In contrast, independent fusion \yg{retains far more of it}, achieving 65.10\% / 68.25\%.

\begin{tcolorbox}[colback=SeaGreen!10!CornflowerBlue!10,colframe=RoyalPurple!55!Aquamarine!100!,title=Summary of RQ3]
Independent fusion outperforms interactive fusion by \textbf{+4.49\%} on VerilogEval-v2 and \textbf{+4.56\%} on ResBench, \yg{confirming the advantage of independent acquisition under prompt-level interaction}. \yg{Exposing execution feedback} induces anchoring bias that reduces reasoning's independent contribution, causing it to degenerate toward execution-only performance.
\end{tcolorbox}

%==============================================================================
%==============================================================================
\subsection{\yg{RQ4: Validity of the Execution Anchor}}
\label{subsec:rq4}

\subsubsection{\yg{Motivation}}
\yg{Clustering and fusion rest on testbenches that {\tool} writes for itself, so the execution channel must be validated rather than trusted as an oracle.
We ask how reliable it is and how the framework behaves where it is not.}

\subsubsection{\yg{Experimental Design}}
\yg{Each benchmark hides a ground-truth testbench used only for measurement, never exposed to the reranker.
Running it on the ten candidates per problem gives a reference label; running the generated testbench gives the signal {\tool}-T actually consumes.
Comparing the two yields a \textit{false-accept} rate (passed by the generated testbench but rejected by the reference), a \textit{false-reject} rate (the converse), the fraction of clusters mixing correct and incorrect candidates, and the AUC of the execution score against the reference label.
We also measure the GPT-5 testbenches used by CodeT-GPT as a non-retrieval reference point.}

\yg{We evaluate by fault detection rather than structural coverage: coverage measures how much of a \textit{single} implementation is exercised, whereas the anchor must separate \textit{several} implementations of one specification.
A testbench with loose expected values can achieve full coverage on a correct module yet return the same verdict for all candidates.
False acceptance is a fault detection rate over the incorrect implementations the nine generators actually produced.}

\begin{table}[htbp]
\centering
\caption{\yg{Quality of the execution anchor against the hidden ground-truth testbenches, aggregated over the nine generators. Lower is better for F-accept, F-reject, and Mixed.}}
\label{tab:anchor}
\resizebox{\columnwidth}{!}{
\begin{tabular}{l|ccccc}
\toprule
\textbf{Testbench source} & \textbf{Compiles} & \textbf{F-accept} & \textbf{F-reject} & \textbf{Mixed} & \textbf{AUC} \\
\midrule
\multicolumn{6}{c}{\yg{\textit{VerilogEval-v2}}} \\
\midrule
\yg{GPT-5 (CodeT-GPT)} & \yg{92.9} & \yg{15.4} & \yg{40.0} & \yg{20.2} & \yg{0.768} \\
\yg{{\tool}-T ($k$=1)} & \yg{38.4} & \yg{8.5} & \yg{69.0} & \yg{30.3} & \yg{0.618} \\
\yg{{\tool}-T ($k$=5)} & \yg{82.3} & \yg{12.8} & \yg{47.3} & \yg{23.8} & \yg{0.743} \\
\midrule
\multicolumn{6}{c}{\yg{\textit{ResBench}}} \\
\midrule
\yg{GPT-5 (CodeT-GPT)} & \yg{58.9} & \yg{5.9} & \yg{59.2} & \yg{25.7} & \yg{0.702} \\
\yg{{\tool}-T ($k$=1)} & \yg{64.3} & \yg{8.0} & \yg{73.7} & \yg{39.4} & \yg{0.571} \\
\yg{{\tool}-T ($k$=5)} & \yg{75.0} & \yg{6.6} & \yg{49.6} & \yg{27.5} & \yg{0.727} \\
\bottomrule
\end{tabular}
}
\end{table}

\subsubsection{\yg{Results}}

\textbf{\yg{The anchor is informative but conservative.}}
\yg{At $k$=5, false acceptance is 12.8\% and 6.6\% while false rejection reaches 47.3\% and 49.6\% (Table~\ref{tab:anchor}): the anchor rarely certifies an incorrect candidate but often fails to certify a correct one (AUC 0.743 and 0.727).
This asymmetry is why {\tool} fuses the score rather than filters on it, since discarding every rejected candidate would remove about half of the correct ones.
Roughly a quarter of the clusters still mix correct and incorrect candidates, which reasoning must separate.}

\textbf{\yg{Retrieval depth and testbench source.}}
\yg{At $k$=1 only 38.4\% of testbenches compile on VerilogEval-v2 and AUC falls to 0.618, which is why the main experiments use $k$=5.
Against the GPT-5 testbenches, retrieval improves every measure on ResBench but not on VerilogEval-v2, where GPT-5 compiles more often and separates candidates slightly better.
Anchor quality thus varies with the source, and the fusion weight $\alpha$ absorbs this variance.}

\textbf{\yg{Behavior under an uninformative anchor.}}
\yg{We split the 1,872 VerilogEval-v2 instances by whether the generated testbench separates the ten candidates.
On the 1,117 instances where it returns one verdict for all, {\tool} reaches 79.43\% vs.\ 78.51\% for {\tool}-R and 74.92\% for {\tool}-T: with no execution signal, the ranking follows reasoning rather than collapsing.
On the remaining 755 it reaches 43.74\% vs.\ 40.31\% and 33.79\%, so the anchor contributes where it carries information.
The first group holds the easier problems (oracle 82.81\% vs.\ 50.07\%), reflecting the absence of anything to separate rather than a testbench limitation alone.}

\begin{tcolorbox}[colback=SeaGreen!10!CornflowerBlue!10,colframe=RoyalPurple!55!Aquamarine!100!,title=Summary of RQ4]
\yg{The anchor false-accepts 6.6\%--12.8\% but false-rejects 47.3\%--49.6\%, favoring fusion over filtering. When it separates no candidates, {\tool} tracks reasoning (79.43\% vs.\ 78.51\%) instead of degrading; when it does separate, it gains most (43.74\% vs.\ 40.31\%).}
\end{tcolorbox}
\section{Discussion}
\label{sec:discuss}

\subsection{Design Choices and Channel Independence}
\label{subsec:design_choices}
A core design principle of {\tool} is the independent acquisition of execution and reasoning signals. 
To achieve this, we employ an asymmetric architecture: LoRA fine-tuning for the reasoning channel ({\tool}-R) and Retrieval-Augmented Generation (RAG) for the execution channel ({\tool}-T).

\textbf{Justification for Asymmetric Design.}
\yg{Table~\ref{tab:ablation_tuning} shows that both channels need the tuned backbone, but they acquire domain knowledge differently.}
For {\tool}-R, LoRA fine-tuning on VeriJudge-47K is highly effective for internalizing judgment capabilities, whereas adding RAG to an untuned model falls short. 
\yg{{\tool}-T reuses that same judgment-tuned adapter and obtains its testbench knowledge from retrieval rather than from weights: an untuned model compiles poorly even with RAG, yet training on VeriTest-53K without retrieval still trails retrieving from it (60.12 vs.\ 62.50 on ResBench). Test structures and assertion formats are thus better supplied as in-context examples than absorbed into parameters, which is what makes the architecture asymmetric.}

\begin{table}[htbp]
\centering
\caption{Ablation on fine-tuning and RAG configurations for {\tool}-R and {\tool}-T. Best results are \textbf{bolded}.}
\label{tab:ablation_tuning}
\resizebox{\columnwidth}{!}{
\begin{tabular}{l|cc}
\toprule
\textbf{Configuration} & \textbf{VerilogEval-v2} & \textbf{ResBench} \\
\midrule
\multicolumn{3}{c}{\cellcolor{gray!15}\textit{Reasoning Channel ({\tool}-R)}} \\
\midrule
Untuned Base Model & 56.12 & 57.13 \\
Untuned + RAG & 59.34 & 60.07 \\
\textbf{Tuned on VeriJudge-47K (Current)} & \textbf{62.77} & \textbf{63.10} \\
\midrule
\multicolumn{3}{c}{\cellcolor{gray!15}\textit{Execution Channel ({\tool}-T)}} \\
\midrule
Untuned + RAG & 55.77 & 58.33 \\
Tuned on VeriJudge-47K (No RAG) & 55.34 & 57.54 \\
Tuned on VeriJudge \& VeriTest (No RAG) & 56.84 & 60.12 \\
\textbf{Tuned on VeriJudge-47K + RAG (Current)} & \textbf{57.98} & \textbf{62.50} \\
\bottomrule
\end{tabular}
}
\end{table}

\textbf{Preserving Channel Independence.}
\yg{A natural concern is whether sharing one base model (and LoRA adapter) introduces coupling.
At inference, {\tool}-R sees only the requirement and candidate code while {\tool}-T sees only the requirement and retrieved examples, so no execution outcome reaches the reasoner.
A shared backbone does not rule out statistical dependence, however: both channels can inherit the same misconception and fail together.
We count this among the sources of correlated error discussed in Section~\ref{subsec:theory_bounds}.}

\subsection{Theoretical Assumptions and Boundaries}
\label{subsec:theory_bounds}
\yg{The analysis in Section~\ref{sec:theory} explains why independent acquisition is a sound default; it does not show that independent fusion beats every interactive design. We discuss below where its assumptions can break.}

\yg{\textbf{Conditional Independence and Its Violation.}
Proposition~\ref{thm:complementarity} uses $P(\mathbf{E}, R \mid Y) \approx P(\mathbf{E} \mid Y) \cdot P(R \mid Y)$.
Since both channels read the same requirement $x$, an under-constrained specification can make them fail together: the \testbench generator verifies one misreading while the reasoner endorses it.
Our ablations show that substantial non-redundant information survives on these benchmarks, but they do not establish exact independence.}

\yg{We measure the violation directly.
On instances where at least one candidate is correct and one is not, both channels select an incorrect candidate on 16.1\% of them, 1.71$\times$ the 9.4\% independence predicts (Table~\ref{tab:joint_failure}); the ratio is 1.73 on ResBench.
Dependence tracks design style: sequential logic roughly doubles per-channel error and more than doubles joint failure, with state machines the worst case, because such specifications commonly leave reset polarity or unenumerated-state behavior implicit.
Joint failure is also inflated by problems that are simply hard for both channels, so these rates bound correlated error rather than isolate a shared misreading; multi-clock designs, absent from both benchmarks, remain an unquantified risk.}

\begin{table}[htbp]
\centering
\caption{\yg{Failure of each channel on VerilogEval-v2, over the 533 instances containing both a correct and an incorrect candidate. Under independence $P(\text{both})$ would equal $P(T)P(R)$, which is 0.094 overall.}}
\label{tab:joint_failure}
\resizebox{\columnwidth}{!}{
\begin{tabular}{l|cccc}
\toprule
\textbf{Stratum} & \textbf{$n$} & \textbf{$P(T)$} & \textbf{$P(R)$} & \textbf{$P$(both)} \\
\midrule
\yg{All} & \yg{533} & \yg{0.381} & \yg{0.248} & \yg{0.161} \\
\yg{Combinational} & \yg{254} & \yg{0.244} & \yg{0.185} & \yg{0.094} \\
\yg{Sequential} & \yg{279} & \yg{0.505} & \yg{0.305} & \yg{0.222} \\
\yg{Finite-state machine} & \yg{155} & \yg{0.484} & \yg{0.374} & \yg{0.258} \\
\bottomrule
\end{tabular}
}
\end{table}

\textbf{Markov Assumption and Theorem Scope.}
\yg{Theorem~\ref{thm:independent} applies the DPI to post-hoc interaction $R' = h(R, \mathbf{E})$, which presumes a Markov chain.
With parameters fixed, the chain holds; it would break if the judge had memorized the reference implementation, since $R'$ could then carry information about $Y$ that neither $R$ nor $\mathbf{E}$ supplies, motivating our contamination screening (Section~\ref{subsec:contamination}).
The scope is also narrower than interaction in general: it says nothing about a judge that retains its independent assessment while exploiting $\mathbf{E}$ separately. Our preference for independent acquisition therefore rests on the RQ3 measurements rather than a general optimality claim.}

\textbf{Faithfulness of Mutual Information.}
\yg{$I(Y; R)$ quantifies verdict utility: a correct Yes/No may follow a flawed chain, so the measure reflects the verdict's value rather than the soundness of the reasoning.
Compiler-in-the-loop filtering and majority voting reduce verdict noise but neither inspects the chain. Our claims are correspondingly about verdict-level information.}

\yg{\textbf{Does the Training Filter Collapse Reasoning into Execution?}
VeriJudge-47K retains only candidates whose teacher verdict agrees with execution (Section~\ref{subsec:eahc-r}).
The two targets differ: the filter uses \emph{ground-truth} testbenches, so the retained label tracks $Y$, whereas $\mathbf{E}$ at inference comes from a \emph{generated} five-case testbench with coverage gaps (Section~\ref{sec:empirical}).
Training toward $Y$ thus does not reduce to imitating $\mathbf{E}$; accordingly, {\tool}-R alone outperforms {\tool}-T alone (62.77 vs.\ 57.98 on VerilogEval-v2), which an approximation of $\mathbf{E}$ could not achieve.
The filter does cost coverage of hard instances: candidates that every teacher misjudged are discarded, leaving the judge weakest where judgment is hardest.}

\yg{\textbf{Cluster Reasoning Aggregation.}
Each cluster inherits the $\max$ of its members' reasoning scores.
Because $\max$ propagates the most confident judgment, one hallucinated \texttt{Yes} can promote a cluster on its own.
Table~\ref{tab:agg_ablation} compares $\max$ against mean and median: $\max$ still ranks first, by under one point.
Averaging suppresses an isolated hallucination but also penalizes a correct cluster with uneven endorsements; the second effect dominates here.
The aggregator choice is in any case second-order: clustering itself recovers +2.33\% and +5.15\% over {\tool}-R (Section~\ref{subsubsec:consistency}), an order of magnitude more than separates $\max$ from median.}

\begin{table}[htbp]
\centering
\caption{\yg{Cluster reasoning aggregation ablation. Average Pass@1 over the nine generators.}}
\label{tab:agg_ablation}
\begin{tabular}{l|cc}
\toprule
\textbf{Aggregation} & \textbf{VerilogEval-v2} & \textbf{ResBench} \\
\midrule
\yg{$\max$ (default)} & \yg{\textbf{65.10}} & \yg{\textbf{68.25}} \\
\yg{Mean} & \yg{64.60} & \yg{67.66} \\
\yg{Median} & \yg{64.32} & \yg{67.26} \\
\bottomrule
\end{tabular}
\end{table}

\subsection{Broader Implications and HDL Specificity}
\label{subsec:implications}

\noindent\textbf{Importance to Software Engineering.}
Verilog defects propagate into physical fabrication, \yg{making them far costlier to fix than software defects, which can still be patched after release}. 
As LLM-based Verilog generation gains industrial adoption without principled quality assurance, the unique challenges of HDLs (e.g., parallel semantics, timing) present critical opportunities for SE techniques like test generation and reranking~\cite{chen2023essence,fang2025lintllm,chen2026qihe}.

\noindent\textbf{Are the Limitations Exclusive to HDLs?}
While poor domain transferability and reasoning hallucination exist in general-purpose languages, they are \yg{\textit{markedly} more severe} in Verilog. 
For example, execution-based methods yield \yg{substantial} gains on Python~\cite{chen2022codet} but struggle on Verilog due to LLMs' inability to generate high-quality hardware testbenches. 
Combined with data scarcity and complex semantics, these exacerbated challenges \yg{motivate} targeted frameworks like {\tool}.

\subsection{\yg{Contamination and Diversity Analysis}}
\label{subsec:contamination}
\yg{Both benchmarks release only a natural-language prompt, a module header, and a hidden testbench; neither publishes a reference implementation.
Overlap can therefore be measured only on the problem side, so we represent each benchmark item by its prompt and header, which is also the query {\tool}-T issues to the retrieval corpus at inference; VeriTest-53K questions state the same two fields and are thus directly comparable.
Each problem is scored by its nearest corpus neighbour under word-level ROUGE-L and LFM2.5-Embedding-350M cosine similarity.}

\yg{No benchmark problem has a neighbour above 0.9 on either measure (Table~\ref{tab:contamination}); the closest reach 0.667 ROUGE-L and 0.862 cosine.
Median lexical overlap is 0.285 and 0.437, whereas cosine exceeds 0.7 for most pairs: both sides are Verilog specifications sharing vocabulary and framing, which is domain relatedness rather than problem-level leakage.}

\begin{table}[htbp]
\centering
\caption{\yg{Nearest-neighbour overlap between benchmark problems (prompt + header) and the VeriTest-53K retrieval corpus.}}
\label{tab:contamination}
\resizebox{\columnwidth}{!}{
\begin{tabular}{l|cccc|cccc}
\toprule
 & \multicolumn{4}{c|}{\textbf{ROUGE-L}} & \multicolumn{4}{c}{\textbf{Embedding cosine}} \\
\textbf{Benchmark} & med. & p90 & max & $\geq$0.9 & med. & p90 & max & $\geq$0.9 \\
\midrule
\yg{VerilogEval-v2 (156)} & \yg{0.285} & \yg{0.400} & \yg{0.667} & \yg{0\%} & \yg{0.750} & \yg{0.800} & \yg{0.862} & \yg{0\%} \\
\yg{ResBench (56)} & \yg{0.437} & \yg{0.621} & \yg{0.764} & \yg{0\%} & \yg{0.730} & \yg{0.813} & \yg{0.850} & \yg{0\%} \\
\bottomrule
\end{tabular}
}
\end{table}

\yg{Exact port signatures (direction, bit width, and name, ignoring clock/reset-only stubs) recur for 52/156 VerilogEval-v2 problems but only 3/56 in ResBench.
Such matches reflect standard HDL skeletons rather than shared problems: an 8-bit input to 8-bit output interface constrains nothing about the function to implement, and the more distinctive interfaces in ResBench almost never recur.}

\yg{Diversity is the other side.
VeriJudge-47K holds 47,375 records over 12,195 distinct problem statements (each paired with multiple candidates); VeriTest-53K pairs one testbench with each of its 53,015 distinct questions, so exact repetitions are already collapsed.
Scoring every distinct statement by its nearest \textit{other} statement in the same corpus gives a median cosine of 0.945 and 0.897, seemingly high but expected in a domain where every item is a Verilog specification.
These values calibrate Table~\ref{tab:contamination}: the closest benchmark problem reaches only 0.862, below the typical intra-corpus similarity.
Benchmark problems thus lie further from our corpora than corpus items lie from each other.}

\yg{Residual overlap is benign: a training record pairs a requirement with a sampled candidate and a Yes/No verdict, so a recurring requirement carries no reference answer.
Pre-training exposure remains outside our control, but every reranker in Section~\ref{sec:eval} scores the same candidates and therefore shares it.}

\subsection{Threats to Validity}
\label{sec:threats}

\noindent\textbf{Internal Validity.}
To avoid the risk of implementation errors, we use unit testing and manual verification on sampled cases.
To reduce randomness from candidate sampling and majority voting, we fix random seeds across all experiments.
To mitigate data contamination, we remove samples with ROUGE-L $> 0.9$ against any benchmark problem from VeriJudge-47K and VeriTest-53K\yg{, and audit the residual overlap in Section~\ref{subsec:contamination}}.

\noindent\textbf{External Validity.}
We evaluate on VerilogEval-v2 (156 problems) and ResBench (56 problems), which may not fully represent industrial-scale designs.
Future work should validate {\tool} on larger proprietary benchmarks.
We focus on Verilog; while the dual-channel design is language-agnostic in principle, \yg{transferring it to another HDL requires rebuilding both corpora and retraining, so we make no claim beyond Verilog. Section~\ref{subsec:rq1} instead characterizes where the method underperforms within its stated scope}.
Additionally, the VeriJudge-47K and VeriTest-53K datasets are curated using LLMs (GPT-4o and GLM-4), which may introduce teacher model biases. We mitigate this by using multi-teacher distillation and strict compiler-in-the-loop verification.

\noindent\textbf{Construct Validity.}
Correctness is defined by execution on ground-truth testbenches, which may miss subtle bugs.
\yg{Section~\ref{subsec:rq4} finds the generated testbenches permissive, leaving about a quarter of clusters mixed.
The fusion weight $\alpha{=}0.6$ is a global constant that cannot adapt to an individual weak testbench; performance is stable over $\alpha \in [0.1, 0.6]$ (Fig.~\ref{fig:alpha_analysis}), and conditioning $\alpha$ on per-problem anchor quality is left to future work.}

\noindent\textbf{Cost and Practicality.}
On a single RTX 4090, per-problem latency is ${\sim}$\textbf{15--20s} (vs.\ ${\sim}2$s greedy), a ${\sim}10\times$ increase dominated by testbench generation (${\sim}1$K tokens), simulation (${\sim}0.5$s), and reasoning calls (${\sim}30$K tokens total).
\yg{Cheaper settings hurt: $n{=}1$ leaves contradictory verdicts at 7.5\%/12.8\% instead of 4.8\%/7.1\% (Section~\ref{subsubsec:consistency}), and $k{=}1$ compiles only 38.4\% of testbenches (Section~\ref{subsec:rq4}).
Our configuration is thus near the diminishing-returns point; remaining latency is better addressed by confidence-based early stopping, quantization, or a smaller distilled judge, which we leave to future work.
Whether the cost is justified depends on the generator: gains reach +18.59 on Open-Coder while the strongest models leave little for any reranker to recover, so the overhead is best spent where the Pass@1-to-Pass@10 gap is wide.}

\section{Conclusion}
\label{sec:conclusion}

We address the Verilog code reranking problem to bridge the gap between Pass@$k$ potential and Pass@1 reality.
Through empirical analysis, we identify two critical limitations: poor domain transferability and reasoning hallucination.
We propose {\tool}, a dual-channel framework that independently acquires execution and reasoning signals \yg{and anchors reasoning to execution behavior so that execution-equivalent candidates are scored alike.
It raises average Pass@1 by over 11\% and 15\% on VerilogEval-v2 and ResBench, ranking first in 15 of the 18 configurations}.
Future work includes extending to other low-resource HDLs, validating on industrial-scale benchmarks, and exploring formal verification as a complementary signal.

\section*{Acknowledgements}

This work was supported by National Key R\&D Program of China (No. 2024YFB4506400).
Guang Yang is also supported by the Postdoctoral Fellowship Program of CPSF under Grant Number GZC20260902.
\bibliographystyle{IEEEtran}
\bibliography{main}

\end{document}